\PassOptionsToPackage{table}{xcolor}
\documentclass[11pt,a4paper]{article}
\usepackage[margin=1in]{geometry}
\usepackage{amsmath,amsthm,amssymb}
\usepackage[dvipsnames]{xcolor}
\usepackage[T1]{fontenc}
\usepackage{mathtools}
\usepackage{multirow}
\usepackage{booktabs}
\usepackage{makecell}
\usepackage{tabularx}
\usepackage{diagbox}
\usepackage{graphicx}
\usepackage{fullpage}
\usepackage{enumitem}
\usepackage{tikz}
\usetikzlibrary{matrix, positioning}
\usepackage{xspace}
\usepackage{url}
\usepackage[linesnumbered,ruled,vlined]{algorithm2e}
\usepackage[colorlinks=true,linktoc=all,linkcolor=blue,citecolor=teal,urlcolor=red,backref=page]{hyperref}
\usepackage[capitalise,noabbrev,nameinlink]{cleveref}
\usepackage{algpseudocode}
\usepackage{appendix}
\usepackage[english]{babel}
\usepackage{bbm}
\usepackage{mathptmx}
\usepackage{microtype}
\usepackage{comment}
\usepackage{placeins}
\usepackage{csquotes}
\usepackage{tablefootnote}
\usepackage{afterpage}
\usepackage{lipsum}
\usepackage{booktabs}
\usepackage{thm-restate}
\usepackage{subcaption}
\usepackage{booktabs}
\usepackage{array}

\DeclareMathAlphabet{\mathcal}{OMS}{cmsy}{m}{n}
\SetMathAlphabet{\mathcal}{bold}{OMS}{cmsy}{b}{n}

\crefname{algocf}{algorithm}{algorithms}  
\Crefname{algocf}{Algorithm}{Algorithms}  
\crefname{enumi}{part}{parts}   
\Crefname{enumi}{Part}{Parts}   
\crefname{footnote}{Footnote}{Footnotes}

\newcolumntype{Y}{>{\centering\arraybackslash}X}

\let\oldtabular\tabular 
\renewcommand{\tabular}{\normalsize\oldtabular}

\theoremstyle{plain}
\newtheorem{theorem}{Theorem}
\newtheorem{lemma}{Lemma}[section]
\newtheorem{claim}[lemma]{Claim}
\newtheorem{corollary}[lemma]{Corollary}

\theoremstyle{definition}
\newtheorem{definition}[lemma]{Definition}

\newtheorem{remark}[lemma]{Remark}
\newtheorem{example}[lemma]{Example}

\newcommand{\abs}[1]{\ensuremath{\mathopen\lvert #1 \mathclose\rvert}}

\newcommand{\np}{\ensuremath{\mathrm{NP}}}
\newcommand{\ps}{\ensuremath{\mathrm{PSPACE}}}
\newcommand{\nps}{\ensuremath{\mathrm{NPSPACE}}}

\newcommand{\ef}{\mathsf{EF}}

\def\gb#1#2{g_{#2,\text{best}}^{#1}}

\def\cw#1#2{{c_{#2,\text{worst}}^{#1}}}
\def\e#1{{\varepsilon}_{#1}}

\def\sm{\setminus}
\def\genvy{{\mathsf{G}}_{\text{envy}}}
\def\G{\mathcal{G}}

\newcommand{\ceq}{\coloneq}
\def\set#1{\ensuremath{\left\{#1\right\}}}

\def\ceil#1{\left\lceil #1 \right\rceil}

\newcommand\restr[2]{{
  \left.\kern-\nulldelimiterspace 
  #1 
  \vphantom{\big|} 
  \right|_{#2} 
  }}

\newcommand{\linkline}[1]{Line~\hyperref[#1]{\ref*{#1}}}

\newcommand{\efconn}{\mathsf{EF1}\text{-}\mathsf{Restoration}}
\newcommand{\efopt}{\mathsf{Optimal}\text{-}\mathsf{EF1}\text{-}\mathsf{Restoration}}
\newcommand{\efxconn}{\mathsf{EFX}\text{-}\mathsf{Restoration}}
\newcommand{\efxopt}{\mathsf{Optimal}\text{-}\mathsf{EFX}\text{-}\mathsf{Restoration}}

\newcommand{\pmr}{\mathsf{Perfect}\text{ }\mathsf{Matching}\text{ }\mathsf{Reconfiguration}}

\newcommand{\gval}{G_{\text{val}}}

\newcommand{\ignore}[1]{}

\allowdisplaybreaks

\title{Fair Division in a Variable Setting}

\author{%
  Harish Chandramouleeswaran\\
   Chennai Mathematical Institute, India\\
   \texttt{harishc@cmi.ac.in} \\
   \and
   Prajakta Nimbhorkar\thanks{Corresponding author. E-mail for correspondence: \texttt{prajakta@cmi.ac.in}}\\
   Chennai Mathematical Institute, India\\
   \texttt{prajakta@cmi.ac.in}
\and
   Nidhi Rathi\\
   University of Warsaw, Poland\\
   \texttt{n.rathi@uw.edu.pl}
}

\begin{document}

\maketitle


\begin{abstract} \label{sec:abstract}

We study the classic problem of fairly dividing a set of indivisible items among a set of agents and consider the popular fairness notion of \emph{envy-freeness up to one item} ($\mathsf{EF1}$). While in reality, the set of agents and items may vary, previous works have studied \emph{static} settings, where no change can occur in the system. We initiate and develop a formal model to understand fair division under a \emph{variable input} setting. Given an arbitrary initial allocation, the goal is to reach an $\mathsf{EF1}$ allocation through item transfers while causing {\em minimal disruption} to the agents. We formally introduce the notion of {\em valid transfers} to achieve minimal disruption and refer to this as the {\em $\mathsf{EF1}\text{-}\mathsf{Restoration}$} problem.

We develop efficient algorithms for $\mathsf{EF1}\text{-}\mathsf{Restoration}$ when agents have identical monotone valuations over items that are either all goods or all chores. Surprisingly, even for identical additive valuations, we prove that it is $\mathrm{NP}$-hard to optimize the number of valid transfers for $\mathsf{EF1}\text{-}\mathsf{Restoration}$. For the stronger problem of $\mathsf{EFX}\text{-}\mathsf{Restoration}$, we show that it is $\mathrm{NP}$-hard to decide whether the problem admits a positive solution for identical additive valuations. We also show that unlike the case of goods and chores, $\mathsf{EF1}\text{-}\mathsf{Restoration}$ may not be possible for mixed manna.

For additive binary valuations, we show that (i) it is $\mathrm{NP}$-hard to decide whether $\mathsf{EF1}\text{-}\mathsf{Restoration}$ is possible, and (ii) it is $\mathrm{NP}$-hard to find the \emph{optimal} number of valid transfers for the same.
We complement this with a positive result for a subclass of additive binary valuations that are defined using {\em multigraphs}, introduced by Christodoulou et al.\ (EC 2023). Here, we present a polynomial-time algorithm for the $\mathsf{EF1}\text{-}\mathsf{Restoration}$ problem when allocations are required to be {\em orientations}.

Finally, for monotone binary valuations, we show that the problem of deciding whether $\mathsf{EF1}\text{-}\mathsf{Restoration}$ is possible is $\mathrm{PSPACE}$-complete. Our results provide an extensive picture of the complexity landscape of the $\mathsf{EF1}\text{-}\mathsf{Restoration}$ problem across various valuation classes.  
\end{abstract}

\tableofcontents


\section{Introduction} \label{sec:intro}

\emph{Fair Division} studies the fundamental problem of dividing a set of indivisible resources among a set of interested parties (often dubbed as \emph{agents}) in a \emph{fair} manner. The need for fairness is inherent in the design of many social institutions and occurs naturally in many real-life scenarios such as inheritance settlements, radio and satellite spectrum allocations, and air traffic management, to name a few \cite{budish2012multi,etkin2007spectrum,moulin2004fair,vossen2002fair,pratt1990fair}. While the first mentions of fair division date back to the Bible and Greek mythology, its first formal study is credited to Steinhaus, Banach, and Knaster in 1948 \cite{steinhaus1948problem}. Since then, the theory of fair division has enjoyed flourishing research from mathematicians, social scientists, economists, and computer scientists alike. The last group has brought a new flavor of questions related to different models of computing \emph{fair} allocations. We refer to \cite{survey2022,brams1996fair,brandt2016handbook,robertson1998cake} for excellent expositions and surveys of fair division.

\emph{Envy-freeness} is one of the quintessential notions of fairness for resource allocation problems that entails every agent to be \emph{happy} with their own share and not prefer any other agent's share to theirs in the allocation, i.e., being \emph{envy-free} \cite{foley1966resource}. This notion has strong existential properties when the resource to be allocated is divisible, like a \emph{cake} \cite{stromquist1980cut,Su1999rental}. But, a simple example of two agents and one valuable item shows that envy-free allocations may not exist for the case of indivisible items. Hence, several variants of envy-freeness have been explored in the literature. \emph{Envy-freeness up to one item} ($\mathsf{EF1}$) is one such popular relaxation \cite{budish2011combinatorial} that entails an allocation to be \emph{fair} when any envy for an agent must go away after removing some item from the envied bundle. When we have an $\mathsf{EF1}$ allocation, we say that the agents are \emph{$\mathsf{EF1}$-happy}, while if an agent's envy does not go away after removing some item from the envied bundle, we refer to it as \emph{$\mathsf{EF1}$-envy}. It is known that $\mathsf{EF1}$ allocations always exist and can be computed efficiently as well \cite{lipton2004approximately}.

Given that the notion of $\mathsf{EF1}$ is tractable, it becomes a tempting problem to explore it in a \emph{changing environment}, i.e., what happens when we have a \emph{variable} input of items and agents, as is natural in many real-world scenarios. This work aims to develop tools for the problem of \emph{restoring fairness} in a varying system while causing \emph{minimal disruption to the existing allocation}. 

We initiate the study of fair division of indivisible items in a changing environment due to variable input. In this work, we go beyond the classic \emph{static} setting of fair division and study an interesting variant where some agent(s) may lose some item(s) in a given $\mathsf{EF1}$ allocation, and therefore, may no longer be $\mathsf{EF1}$-happy. Another such scenario is when some new agent(s) enter the system with empty bundles to begin with, and naturally $\mathsf{EF1}$-envy some of the existing agents. Yet another scenario is when the valuations may change over time.

Given an arbitrary allocation, the goal is to redistribute the items by causing {\em minimal disruption} to the existing allocation and reach an $\mathsf{EF1}$ allocation. To do so, we consider a sequence of \emph{item transfers} such that it does not create any new $\mathsf{EF1}$-envy among the agents. We refer to such transfers that cause minimal disruption to the system as \emph{valid transfers}. Our aim is to modify the input allocation by performing a sequence of valid transfers between agents and eventually reach an $\mathsf{EF1}$ allocation.
 
Note that a natural approach is to recompute an $\mathsf{EF1}$ allocation from scratch (say, by using the algorithm in \cite{lipton2004approximately}), but this may lead to extensive (and possibly unnecessary) changes across all bundles. Such global recomputation is often undesirable, as it disrupts existing assignments. Instead, it is important (and desirable) to restore fairness while causing minimal disruption to the existing allocation. This is motivated by real-world applications where allocations may evolve over time. For example, in an organization, new employees may join or tasks may get completed and removed from the system. A gradual redistribution of workload, while minimizing reshuffling and maintaining fairness is essential for the smooth functioning of the system and also for employee satisfaction. Another example is where machines go offline due to technical faults and we need to gradually redistribute capacity without creating huge unfair changes.

A central difficulty in restoring fairness is that arbitrary item transfers may eliminate existing EF1 violations while simultaneously creating new ones. We study the problem of $\efconn$: given an initial allocation, can one reach an $\mathsf{EF1}$ allocation via a sequence of small, controlled changes?  We capture this requirement through the notion of \emph{valid transfers}, which constrain how items may be reassigned along the path to fairness. An item transfer is \emph{valid} if it does not create any new $\mathsf{EF1}$-envy among agents.

The concept of valid transfers is inspired by the work of Igarashi et al.\ ~\cite{igarashi2024reachability} where they consider the problem of deciding if two given $\mathsf{EF1}$ allocations can be reached from one another via a sequence of \emph{exchanges} such that every intermediate allocation remains $\mathsf{EF1}$. We extend this framework to our setting and define \emph{valid transfers} such that no intermediate allocation creates any new $\mathsf{EF1}$-envy in the system. 

We remark that a different non-static setting of \emph{online} fair division has also been studied, in which items or agents arrive and depart over time and the aim is to dynamically build (from scratch) a fair division against an uncertain future; see the survey \cite{online15}. Another line of research is to explore fair division protocols that remain consistent with 
\emph{population monotonicity} and \emph{resource monotonicity}, in which the set of agents or items may vary \cite{chakraborty2021picking}.

\paragraph{Our Contribution:}
In this work, we develop a formal model of fair division under variable input. Given an input allocation, we study the problem of $\efconn$ where the goal is to reach an $\mathsf{EF1}$ allocation via a sequence of \emph{valid transfers} of items, i.e., no item-transfer creates new $\mathsf{EF1}$-envy in the system. We say that $\efconn$ is possible, if there exists a sequence of valid transfers (starting from the input allocation) that leads to an $\mathsf{EF1}$ allocation. We study the related optimization problem, $\efopt$, that determines the minimum number of valid transfers required for $\efconn$. Our main results are listed below, and summarized in \cref{tab:main-results-summary}.

\begin{enumerate}
\item In \cref{sec:identical}, we consider 
the $\efconn$ problem for \emph{identical} valuations. 
\begin{itemize}
    \item In \cref{thm:identical-general}, we show that $\efconn$ is possible  for \emph{identical monotone valuations} and design a polynomial-time algorithm (\Cref{algo:identical-general}) for the same. Surprisingly, we show that $\efopt$ is $\mathrm{NP}$-hard even for identical additive valuations (\cref{thm:k-valid-transfers-EF1-NP-hard-identical-additive}). 
    \item We extend these results in \cref{subsec:identical-chores-general} and prove analogous results for the case of chores as well (\cref{thm:identical-chores-general}). However, $\efconn$ may not always be possible for the case of \emph{mixed manna}, i.e., a mix of goods and chores, even for identical additive valuations. We present such a negative instance (\cref{example:identical-ef1-mixed-counterexample}) in \cref{subsec:ef1-mixed-case}.
\end{itemize}

\item In \cref{sec:additive-binary}, we consider the $\efconn$ problem for \emph{additive binary} valuations.
\begin{itemize}
    \item In \cref{subsec:additive-binary-NP-hard}, we prove (in \cref{thm:EF1-Restoration-NP-hard-additive-binary}) that the following problems are $\mathrm{NP}$-hard for instances with \emph{additive binary} valuations: 
(i) deciding whether $\efconn$ is possible, and (ii) $\efopt$.

\item In \cref{subsec:graphical}, we discuss $\efconn$ for a subclass of additive binary valuations -- namely, the class of \emph{graphical valuations} introduced by Christodoulou et al.\ \cite{christodoulou2023fair}. In \Cref{thm:binary-arb}, we show that, for additive binary valuations over multigraphs, $\efconn$ is always possible in polynomial time, when the allocations are required to be \emph{orientations}. Moreover, the number of valid transfers made by the algorithm is optimal up to an additive factor independent of $n$. See \cref{sec:prelims} for the formal definitions of graphical valuations and an orientation.
\end{itemize}
   
\item In \cref{sec:pspace}, we prove that $\efconn$ is $\ps$-complete for \emph{monotone binary} valuations (\cref{thm:pspace}).

\item Finally, in \cref{sec:identical-efx}, we discuss the analogous problem of $\mathsf{EFX}$\text{-}$\mathsf{Restoration}$, for the stronger $\mathsf{EFX}$ fairness notion \cite{CaragiannisKMPS19}. We show that restoring $\mathsf{EFX}$ is not always possible via valid transfers or exchanges even for identical additive valuations (\cref{example:identical-efx-counterexample}). Furthermore, we prove (in \cref{thm:EFX-Restoration-weakly-NP-hard-identical-additive}) that the following problems are $\mathrm{NP}$-hard in the above setting: 
(i) deciding whether $\efconn$ is possible, and (ii) $\efopt$. 
\end{enumerate}

\begin{table}[htbp]
\centering
\small
\setlength{\tabcolsep}{2pt}
\renewcommand{\arraystretch}{1.15}

\begin{tabular}{|>{\raggedright\arraybackslash}m{3.8cm}|
                >{\centering\arraybackslash}m{3.4cm}|
                >{\centering\arraybackslash}m{4.6cm}|
                >{\centering\arraybackslash}m{3.6cm}|}
\hline
\multicolumn{1}{|>{\columncolor{gray!20}\centering\arraybackslash}m{3.8cm}|}{\textbf{Valuation class}}
&
\multicolumn{1}{>{\columncolor{gray!20}\centering\arraybackslash}m{3.4cm}|}{\textbf{\(\efconn\)}}
&
\multicolumn{1}{>{\columncolor{gray!20}\centering\arraybackslash}m{4.6cm}|}{\textbf{\(\efopt\)}}
&
\multicolumn{1}{>{\columncolor{gray!20}\centering\arraybackslash}m{3.6cm}|}{\textbf{Reference}}
\\
\hline

\cellcolor{green!10} \makecell{Identical monotone \\ \emph{(goods$^*$ / chores$^{**}$)}}
&
\cellcolor{blue!8} \makecell{\(\mathrm{P}\) \\ (\cref{thm:identical-general})}
&
\cellcolor{blue!8} strongly \(\mathrm{NP}\)-complete
&
\cellcolor{orange!10} \cref{subsec:ef1-restoration-identical-monotone}
\\
\hline

\cellcolor{green!10} \makecell{Identical additive \\ \emph{(goods / chores$^{**}$)}}
&
\cellcolor{blue!8} \(\mathrm{P}\)
&
\cellcolor{blue!8} strongly \(\mathrm{NP}\)-complete (\cref{thm:k-valid-transfers-EF1-NP-hard-identical-additive})
&
\cellcolor{orange!10} \cref{subsec:NP-hard-opt-identical-additive}
\\
\hline

\cellcolor{green!10} \makecell{Identical additive \\ \emph{(mixed manna)}}
&
\cellcolor{red!15} \textsc{Open}
&
\cellcolor{red!15} \textsc{Open}
&
\cellcolor{orange!10} \cref{subsec:ef1-mixed-case}$^{\dagger}$
\\
\hline

\cellcolor{green!10} \makecell{Additive binary}
&
\cellcolor{blue!8} \makecell{\(\mathrm{NP}\)-hard \\ (\cref{thm:EF1-Restoration-NP-hard-additive-binary})}
&
\cellcolor{blue!8} \makecell{\(\mathrm{NP}\)-hard \\ (\cref{thm:EF1-Restoration-NP-hard-additive-binary})}
&
\cellcolor{orange!10} \cref{subsec:additive-binary-NP-hard}
\\
\hline

\cellcolor{green!10} \makecell{Additive binary \\ on multigraphs}
&
\cellcolor{blue!8} \makecell{\(\mathrm{P}\) \\ (\cref{thm:binary-arb})}
&
\cellcolor{blue!8} \makecell{\(\mathrm{P}^{\dagger \dagger}\) \\ (\cref{thm:binary-arb})}
&
\cellcolor{orange!10} \cref{subsec:graphical}
\\
\hline

\cellcolor{green!10} \makecell{Monotone binary}
&
\cellcolor{blue!8} \(\mathrm{PSPACE}\)-complete (\cref{thm:pspace})
&
\cellcolor{blue!8} \(\mathrm{NP}\)-hard
&
\cellcolor{orange!10} \cref{sec:pspace}
\\
\hline

\cellcolor{green!10} \makecell{Identical additive (\(\mathsf{EFX}\)) \\ \emph{(goods / chores$^{**}$)}}
&
\cellcolor{blue!8} \makecell{weakly \(\mathrm{NP}\)-hard \\ (\cref{thm:EFX-Restoration-weakly-NP-hard-identical-additive})}
&
\cellcolor{blue!8} \makecell{weakly \(\mathrm{NP}\)-hard \\ (\cref{thm:EFX-Restoration-weakly-NP-hard-identical-additive})}
&
\cellcolor{orange!10} \cref{sec:identical-efx}
\\
\hline
\end{tabular}

\vspace{0.4ex}
\begin{minipage}{0.97\linewidth}
\footnotesize
\emph{Note:} $^*$For identical monotone valuations, \cref{thm:identical-general} yields a valid restoration sequence with at most \(m\) transfers. 

$^{**}$The analogous results also hold for the \emph{chores} setting, with similar proofs (see \cref{thm:identical-chores-general}, \cref{rem:opt-ef1-chores-np-hard}, and \cref{rem:efx-chores}). 

$^{\dagger}$For the \emph{mixed manna} setting, \cref{example:identical-ef1-mixed-counterexample} presents an instance where $\efconn$ is not possible.

$^{\dagger \dagger}$The number of valid transfers made is optimal \emph{up to an additive factor independent of $n$}.
\end{minipage}

\caption{Summary of our main results}
\label{tab:main-results-summary}
\end{table}

\paragraph{Further Related Work:} As mentioned above, the closest work to ours is by Igarashi et al.\ ~\cite{igarashi2024reachability}, in which they consider the problem of deciding if two given $\mathsf{EF1}$ allocations can be reached from one another via a sequence of \emph{exchanges} such that every intermediate allocation remains $\mathsf{EF1}$. They prove that reachability is guaranteed for two agents with identical or binary valuations, as well as for any number of agents with identical binary valuations. Furthermore, they show that two $\mathsf{EF1}$ allocations may not be reachable from one another even in the case of two agents, and deciding reachability is $\ps$-complete in general.

Ito et al.\ ~\cite{ItoKKKNOO23} also consider a similar problem of reachability (to a given target allocation) via exchanges for dichotomous valuations in the setting where communication among the agents is limited, and is represented by an edge. Here, a pair of agents can exchange their items only if the two agents are connected by an edge and they approve the items they receive. They prove that this problem is $\ps$-complete even when the communication graph is complete, but it can be solved in polynomial time if the input graph is a tree.

There are two related papers that have appeared since the AAMAS 2025 conference version of this work \cite{ChandramouleeswaranNR25}. Given an initial allocation of goods, Yuen et al.\ \cite{yuen2024reforming} consider the problem of \enquote{reforming} it via a sequence of \emph{exchanges} to attain $\mathsf{EF1}$. The exchange operation restricts the resulting $\mathsf{EF1}$ allocation to have the same {\em size vector} as the initial allocation. In contrast, in our work, the operation is a \emph{valid transfer of a single item}. The operation of the \emph{exchange of a pair of items} has also been used in \cite{igarashi2024reachability, ItoKKKNOO23, yuen2024reforming}. Dong et al.\ \cite{dong2025reconfiguring} study the problem of reconfiguring \emph{proportional committees}: given two proportional committees, they discuss the existence of a transition path that consists only of proportional committees, where each transition involves replacing one candidate with another candidate.

Another well-studied notion of fairness in the discrete setting is that of \emph{envy-freeness up to any item} ($\mathsf{EFX}$)~\cite{CaragiannisKMPS19}. We say that agent $i$ \emph{strongly envies} agent $j$ if there exists an item $g \in X_j$ such that the envy of $i$ toward $j$ persists even after removing $g$ from $X_j$. An allocation is \emph{EFX} if no agent strongly envies another agent. $\mathsf{EFX}$ is arguably the most popular one in discrete fair division, and is stronger than $\mathsf{EF1}$ (see \cite{Procaccia20} for an excellent survey). The existence of EFX allocations is known only in very restricted settings; e.g., for a limited number of agents \cite{plaut2020almost,chaudhury2020efx,berger2021almost,AACGMM25}, limited number of items \cite{amanatidis2020multiple,mahara2024extension}, limited number of distinct valuations \cite{maharaExtensionAdditive24,hvetalEFXExists24}, and special valuation classes \cite{plaut2020almost,halpern2020fair,amanatidis2021maximum}. Whether EFX allocations exist for additive valuations remains a major open problem in fair division theory. A recent work \cite{akrami2026counterexampleefxnge} has shown non-existence of $\mathsf{EFX}$ for three agents with for monotone valuations.

As mentioned earlier, several non-static models of fair division have been studied. One such model considers the online nature of either items or agents arriving or departing over time, in which one needs to dynamically allocate the items in a fair manner against an uncertain future.
There are various works that tackle different dimensions of the problem: (i) resource being divisible \cite{walsh2011online, VardiPF22}, or indivisible \cite{aleksandrov2017pure,kash2014no}, (ii) resources being fixed and agents arriving over time \cite{walsh2011online}, or agents being fixed and resources arriving over time \cite{online15, he2019achieving, BenadeHP25}, or both arriving online \cite{mattei2017mechanisms}, and (iii) whether mechanisms are informed \cite{walsh2011online} or un-informed \cite{he2019achieving, VardiPF22, BenadeHP25} about the future.

Another line of work studies resource- and population-monotonicity, where the set of agents or items may change. The goal is to construct fair division protocols in which the utility of all participants change in the same direction - either all of them are better off (if there is more to share, or fewer to share among), or they are all worse off (if there is less to share, or more to share among); see \cite{chakraborty2021picking,segal2018resource,segal2019monotonicity}.

The notion of fairness considered in our work as well as in the results mentioned above is envy-freeness up to one item ($\mathsf{EF1}$). Since its introduction by Budish \cite{budish2011combinatorial}, it has been extensively studied in the literature in various settings \cite{CISZ21, BarmanKV2018, halpern2020fair}, to mention a few. For goods, a polynomial-time algorithm to compute an $\mathsf{EF1}$ allocation was given by Lipton et al.\ \cite{lipton2004approximately}. For chores, a polynomial-time algorithm to compute $\mathsf{EF1}$ allocations for monotone valuations was given by Bhaskar et al.\ \cite{BhaskarSV2020}. For a mix of goods and chores, Aziz et al.\ \cite{AzizCIW19} defined an analogous notion of $\mathsf{EF1}$ and gave a polynomial-time algorithm to compute an $\mathsf{EF1}$ allocation in the case of additive valuations. 

As mentioned previously, the class of \emph{graphical} valuations models the setting in which every item is valued by at most two agents, and can therefore be seen as an edge between them. This model was introduced in \cite{christodoulou2023fair} in the context of \emph{envy-freeness up to any item} ($\mathsf{EFX}$) allocations. They proved the existence of $\mathsf{EFX}$ allocations when the underlying graph is simple. This work has motivated studies on $\mathsf{EFX}$ and $\mathsf{EF1}$ \emph{orientations} and allocations in the graphical setting -- for instance, Zhou et al.\ \cite{landscape24} studied the mixed manna setting with both goods and chores, and proved that determining the existence of $\mathsf{EFX}$ orientations for agents with additive valuations is $\np$-complete, and provided certain tractable special cases. Zeng et al.\ \cite{zeng2024structure} related the existence of $\mathsf{EFX}$ orientations and the chromatic number of the underlying graph. Deligkas et al.\ \cite{deligkas2024ef1efxorientations} showed that $\mathsf{EF1}$ orientations always exist for monotone valuations, and can be computed in pseudopolynomial time. We refer the readers to the excellent surveys by Biswas et al.\ ~\cite{Biswas2023} on fair division under such \enquote{structured set constraints}, and by Amanatidis et al.\ ~\cite{survey2022} and Guo et al.\ ~\cite{guo2023survey} on various notions of fairness in the settings of both goods and chores.


\section{Notation and Preliminaries} \label{sec:prelims}

For a positive integer $k$, we write $[k]$ to denote the set $\{1,2, \ldots,k\}$, and we write $2^S$ to denote the power set of a set $S$.

Consider a set $\G$ of $m$ \emph{indivisible} items that needs to be \emph{fairly} allocated to a set $[n]$ of $n$ agents. The preferences of an agent $i \in [n]$ over these items is specified by a valuation function $v_i \colon 2^{\G} \to {\mathbb{R}}^+ \cup \set{0}$, i.e., $v_i(S)$ denotes the value agent $i \in [n]$ associates to the set of items $S \subseteq \G$. An \emph{allocation} $X = (X_1,\ldots,X_n)$ is a partition of the set $\G$ among the $n$ agents, where $X_i$ corresponds to the \emph{bundle} assigned to agent $i \in [n]$. The quantity $v_i(X_i)$ is also referred to as the \emph{utility} of agent $i$ in the allocation $X$. The goal is to construct `fair', i.e., \emph{envy-free up to one item} ($\mathsf{EF1}$) allocations (see \cref{def:ef1}). The tuple $\mathcal{I}=\left<[n], \G, \set{v_i}_{i \in [n]} \right>$ is called a \emph{fair division instance}. For an item $g \in \G$, we denote $v_i(\set{g})$ simply by $v_i(g)$ for a cleaner presentation.

A valuation function $v$ is said to be \emph{monotone} if $v(S) \leq v(T)$ for all $S, T \subseteq \G$ with $S \subseteq T$. $v$ is said to be \emph{binary} if, for each $S \subseteq \G$, and for each $g \in \G$, we have $v(S \cup g) - v(S) \in \set{0,1}$, i.e., if each item $g \in \G$ adds a marginal value of either $0$ or $1$ to any bundle. We also study a natural subclass of monotone valuations, namely \emph{additive} valuations. $v$ is said to be \emph{additive} if, for all $S \subseteq \G$, we have $v(S) = \sum_{g \in S} v(g)$. Observe that an additive valuation is binary if and only if $v(g) \in \set{0,1}$, for each $g \in \G$.

In this work, we also study the class of \emph{graphical valuations}, introduced by Christodoulou et al. \cite{christodoulou2023fair}. These are inspired by geographic contexts, where agents only care about resources nearby and show no interest in those located farther away.

\begin{definition}[Graphical Valuations]\label{def:graphical}
Consider a fair division instance over a graph $\gval = (V,E)$ where vertices correspond to agents and the edges correspond to items. Here, each vertex $i \in V$ has valuation $v_i$ such that $v_i(g)>0$ if and only if the edge $g$ is incident to agent $i$ in $\gval$. For each edge $(i,j) \in E$, we denote by $E(i,j)$ the set of edges (i.e., items) that are valued positively by the agents $i$ and $j$.
\end{definition}

In this work, we consider \emph{multi-graphical valuations} that are additive and binary, i.e., multiple items may be valued (at $1$) by the same pair of agents. We refer to these simply as graphical valuations in the rest of the paper.

Let us now define the concepts of envy and $\mathsf{EF1}$-envy, in order to define the fairness notion of $\mathsf{EF1}$ that is the focus of this work.

\begin{definition}[Envy and $\mathsf{EF1}$-envy]\label{def:ef1-envy}
Agent $i \in [n]$ is said to \emph{envy} agent $j \in [n]$ in an allocation $X$ if $i$ values $j$'s bundle strictly more than her own bundle, i.e., $v_i(X_i) < v_i(X_j)$. Furthermore, we say $i$ \emph{EF1-envies} $j$ if $v_i(X_i) < v_i(X_j \sm g)$ for all $g \in X_j$, i.e., $i$ continues to envy $j$ even after virtually removing \emph{any} one item from $j$'s bundle. 
We say $i$ is \emph{EF1-happy} in $X$ if she does not $\mathsf{EF1}$-envy any other agent (otherwise, $i$ is \emph{EF1-unhappy} in $X$).
\end{definition}

We now define the fairness notion of \emph{envy-freeness up to one item} ($\mathsf{EF1}$).

\begin{definition}[Envy-freeness up to one item ($\mathsf{EF1}$) \cite{budish2011combinatorial}]\label{def:ef1}
An allocation $X$ is said to be $\mathsf{EF1}$ if every agent is $\mathsf{EF1}$-happy in $X$.
\end{definition}

We introduce the following notation to quantify the \enquote{amount of $\mathsf{EF1}$-envy} between a pair of agents in an allocation.

\begin{definition}[Amount of $\mathsf{EF1}$-envy]\label{def:amount-of-ef1-envy}
Let $X=(X_1,\ldots,X_n)$ be an allocation. The amount of $\mathsf{EF1}$-envy that agent $i$ has towards agent $j$ in $X$, denoted by $\e{i\to j}(X)$, is defined as
\[
\e{i\to j}(X) \coloneq \underset{g \in X_j}{\min} \set{v_i(X_j\sm g)-v_i(X_i)}.
\]
Note that $X$ is an $\mathsf{EF1}$ allocation if and only if $\e{i \to j}(X)\leq 0$ for each $i,j\in [n]$.
 \end{definition}

A special type of allocation called an \emph{orientation} has been studied recently, and was introduced by Christodoulou et al. \cite{christodoulou2023fair} while studying $\mathsf{EFX}$ allocations in the context of graphical valuations. Orientations were first studied in the $\mathsf{EF1}$ setting by Deligkas et al. \cite{deligkas2024ef1efxorientations}.

\begin{definition} [Orientation] \label{def:orientation}
An allocation $X$ is an \emph{orientation} if every item is allocated to an agent who has positive marginal value for it.
\end{definition}

An orientation is a type of \emph{non-wasteful allocation}. Note that orientations for graphical valuations with graph $\gval$, are equivalent to assigning a direction to each edge in $\gval$, such that the edge is directed towards the agent who owns the item in a particular allocation.

We define the concept of the \emph{envy graph} of an allocation, that we crucially use to solve $\efconn$ in the case of graphical valuations.

\begin{definition} \label{def:envy-graph}
For an allocation $X = (X_1,\ldots,X_n)$, its \emph{envy graph} $\genvy(X)$ is a directed graph with agents appearing as vertices. We add a directed edge from agent $i$ to agent $j$ if $i$ envies $j$, i.e., if $v_i(X_i) < v_i(X_j)$.
\end{definition}


\section{Our Problem and Contribution}\label{sec:results}

In this section, we define the problem and state our results formally. 
\subsection{Problem Setup} \label{sec:setup}

Given an allocation, the goal of this work is to reach an $\mathsf{EF1}$ allocation by initiating a sequence of \emph{valid transfers}, which we now define.

\begin{definition}[Valid Transfer]\label{def:valid-transfer}
For an allocation $X$, consider a single item transfer from an agent to another, resulting in allocation $Y$. Then, this transfer is said to be \emph{valid} if for all $i,j \in [n]$, either $\e{i \to j}(Y) \leq 0$, or $\e{i \to j}(Y) \leq \e{i \to j}(X)$. In other words, no agent $\mathsf{EF1}$-envies another agent in $Y$ any more than they did in $X$. If $X$ is an orientation, we require that the allocation resulting from a valid transfer is also an orientation.
\end{definition}

Some of our results also discuss \emph{valid exchanges}. An exchange of a pair of items between two agents is said to be \emph{valid}, if there is no new $\mathsf{EF1}$-envy created after the exchange.

An interesting special case of a mildly disrupted fair allocation is one in which exactly one agent is $\mathsf{EF1}$-unhappy. Without loss of generality, we will always assume that agent $1$ is the $\mathsf{EF1}$-unhappy agent (if one exists). Many of our results have important consequences when restricted to this special case.

\begin{definition}[Near-$\mathsf{EF1}$ allocation]\label{def:near-ef1}
An allocation $X$ is said to be \emph{near-$\mathsf{EF1}$} if every agent is $\mathsf{EF1}$-happy in $X$ except one \emph{fixed} agent, who we refer to as the \emph{unhappy agent}. We let agent $1$ denote the unhappy agent, without loss of generality.
\end{definition}

We formally state the problem of fair division in a variable setting that forms the focus of this work.

\begin{restatable}{problem}{ef1restoration}[{\upshape$\efconn$}]
\label{prob:ef1-restoration}
Given a fair division instance and an input allocation $X$, determine whether it is possible to reach an $\mathsf{EF1}$ allocation (from $X$) by a sequence of \emph{valid transfers}.
\end{restatable}

If such a sequence of valid transfers that restores $\mathsf{EF1}$ always exists for a valuation class, we say that \emph{$\efconn$ is possible} for the given class. Next, we  define the following natural optimization problem corresponding $\efconn$.

\begin{restatable}{problem}{ktransfers}[{\upshape$\efopt$}]
\label{prob:k-valid-transfers-EF1}
   Given a fair division instance and an allocation $X$, determine the minimum number of valid transfers required to transform $X$ into an $\mathsf{EF1}$ allocation.
\end{restatable}

We are now in a position to state our main results.


\subsection{Our Main Results} \label{sec:contributions}

We begin by studying identical valuations, and prove that $\efconn$ is possible. Formally, we prove the following (in \cref{subsec:ef1-restoration-identical-monotone}).

\begin{restatable}{theorem}{thmidentical}
\label{thm:identical-general}
    Given a fair division instance with identical monotone valuations over goods and an input allocation $X$, $\efconn$ is possible. Moreover, an $\mathsf{EF1}$ allocation can be obtained after making at most $m$ valid transfers.
\end{restatable}

An analogous result (\cref{thm:identical-chores-general}) also holds for the case in which all the items are \emph{chores}, i.e., if every item has a negative marginal value for all the agents. However, a positive result does not hold in the \emph{mixed manna} setting with both goods and chores (see \cref{subsec:ef1-mixed-case}).
 
In \cref{subsec:NP-hard-opt-identical-additive}, we show that the problem $\efopt$ (\cref{prob:k-valid-transfers-EF1}) is strongly $\mathrm{NP}$-complete, even for identical additive valuations.

\begin{restatable}{theorem}{thmidenticaloptimal}
\label{thm:k-valid-transfers-EF1-NP-hard-identical-additive}
$\efopt$ is strongly $\mathrm{NP}$-complete even for identical additive valuations.
\end{restatable}

$\efopt$ is also strongly $\mathrm{NP}$-complete when all the items are chores. The proof is analogous to that of \cref{thm:k-valid-transfers-EF1-NP-hard-identical-additive}; see \cref{rem:opt-ef1-chores-np-hard}.

Next, in \cref{sec:additive-binary}, we consider $\efconn$ when all the agents have additive binary valuations and show the following result.

\begin{restatable}{theorem}{theorembinary}
\label{thm:EF1-Restoration-NP-hard-additive-binary}
Given a fair division instance with additive binary valuations over goods and an input allocation that is near-$\mathsf{EF1}$, $\efconn$ may not always be possible. In particular, the following problems are $\mathrm{NP}$-hard:
\begin{enumerate}[label={(\roman*)}]
    \item deciding whether $\efconn$ is possible;
    \item $\efopt$.
\end{enumerate}
\end{restatable}

In \cref{subsec:graphical}, we identify an interesting subclass of additive binary valuations for which $\efconn$ is possible, and there is a polynomial-time algorithm to do so optimally. Formally, we study $\mathsf{EF1}$ orientations in the setting of additive binary valuations over multigraphs, and prove the following.

\begin{restatable}{theorem}{thmgraphicalbinary}
\label{thm:binary-arb}
Given a fair division instance with additive binary valuations on multigraphs, and an input orientation, the $\efconn$ problem always admits a polynomial-time computable positive solution (\Cref{algo:identical-general}) that is an orientation.
\end{restatable}

Next, in \cref{sec:pspace}, we prove that $\efconn$ is $\ps$-complete in general, even for a subclass of binary monotone valuations in which every subset of the set of items is valued at either $0$ or $1$, and even when the input allocation is near-$\mathsf{EF1}$ (see \cref{def:near-ef1}).

\begin{restatable}{theorem}{pspace}
\label{thm:pspace}
The $\efconn$ problem is $\ps$-complete for monotone binary valuations.
\end{restatable}

In fact, \cref{thm:pspace} holds even if \emph{valid exchanges} are allowed, apart from valid transfers. We present a reduction from the $\ps$-complete problem of $\pmr$ \cite{BonamyBHIKMMW19} (see \cref{def:perfect-matching-reconfiguration}).

Finally, in \cref{sec:identical-efx}, we consider the analogous $\efxconn$ problem, in which a stronger notion of a valid transfer is used -- there must be no new $\mathsf{EFX}$-envy in the system after each such transfer. We prove that, unlike the case of $\mathsf{EF1}$, $\efxconn$ may not always be possible even for identical additive valuations. Moreover, it is weakly $\mathrm{NP}$-hard to (i) decide if $\efxconn$ is possible, and (ii) optimize the number of valid transfers to reach $\mathsf{EFX}$.

\begin{restatable}{theorem}{thmidenticalefx}
\label{thm:EFX-Restoration-weakly-NP-hard-identical-additive}
Consider fair division instances with identical additive valuations over items (all goods or all chores) and a near-$\mathsf{EFX}$ input allocation. Then, the following problems are weakly $\mathrm{NP}$-hard:
\begin{enumerate}[label={(\roman*)}]
    \item deciding whether $\efxconn$ is possible;
    \item $\efxopt$.
\end{enumerate}
\end{restatable}

Our results are summarized in \cref{tab:main-results-summary}.


\section{\texorpdfstring{$\efconn$}{EF1-Restoration} for Identical Monotone Valuations} \label{sec:identical}

In this section, we discuss the $\efconn$ problem for the case in which all the agents have identical monotone valuations.
In \cref{subsec:ef1-restoration-identical-monotone}, we give a constructive proof (via \Cref{algo:identical-general}) that $\efconn$ is possible in this case, thereby establishing \cref{thm:identical-general}.
In \cref{subsec:NP-hard-opt-identical-additive}, we consider the problem of optimizing the number of valid transfers made in order to restore $\mathsf{EF1}$, when the valuation is additive. We prove that this problem is strongly $\mathrm{NP}$-complete (see \hyperref[thm:k-valid-transfers-EF1-NP-hard-identical-additive-restated]{\Cref*{thm:k-valid-transfers-EF1-NP-hard-identical-additive}}).


\subsection{A Greedy Paradigm for \texorpdfstring{$\efconn$}{EF1-Restoration}} \label{subsec:ef1-restoration-identical-monotone}

In this section, we design a greedy polynomial-time algorithm for $\efconn$ when agents have identical monotone valuations. To this end, we first describe and analyze (in \cref{subsubsec:identical-basic}) our algorithm for the special case where the input allocation is \emph{near-$\mathsf{EF1}$}, i.e., exactly one agent is $\mathsf{EF1}$-unhappy. The case of an arbitrary input allocation is similar, and is discussed in \cref{subsubsec:identical-general}. A similar algorithm also works when all the items are \emph{chores}, and is discussed in \cref{subsec:identical-chores-general}. However, $\efconn$ may not be possible when the set of items consists of a mix of goods and chores -- the case of \emph{mixed manna}. A negative instance (\cref{example:identical-ef1-mixed-counterexample}) is presented in \cref{subsec:ef1-mixed-case}.

\subsubsection{\texorpdfstring{$\efconn$}{EF1-Restoration} for Near-\texorpdfstring{$\mathsf{EF1}$}{EF1} Allocations} \label{subsubsec:identical-basic}

We begin by describing our $\efconn$ algorithm (\Cref{algo:identical-basic}) when the input allocation $X$ is \emph{near-$\mathsf{EF1}$} wherein exactly one agent (say, agent $1$) is $\mathsf{EF1}$-unhappy. 
 At each step of \Cref{algo:identical-basic}, an agent $i \neq 1$ is chosen who would be \emph{least affected} by transferring her most valuable item (from her current bundle) to agent $1$. We denote this item of highest marginal value by $\gb{X}{i}$ and define it as $\gb{X}{i} \ceq \underset{g \in X_i}{\arg\min} \set{v(X_i \sm g)}$. The crux of our proof is to show that the above transfer is \emph{valid} (\cref{def:valid-transfer}). 

\begin{algorithm}[htbp]
\caption{\texorpdfstring{$\efconn$}{EF1-Restoration} of a near-$\mathsf{EF1}$ allocation for identical valuations} \label{algo:identical-basic}
\SetAlgoLined
\SetKwInOut{Input}{Input}
\SetKwInOut{Output}{Output}
\SetKwFunction{FMain}{EF1\_Restorer\_Basic}
\SetKwProg{Fn}{}{:}{}
\SetKw{KwTo}{to}
\SetKw{KwRet}{return}
\DontPrintSemicolon

\Input{\text{ } A fair division instance $\mathcal{I}=\left<[n], \G, v \right>$ with an identical monotone valuation $v$ and a near-$\mathsf{EF1}$ allocation $X = (X_1,\ldots,X_n)$ with agent $1$ being unhappy.}
\Output{\text{ } An $\mathsf{EF1}$ allocation in $\mathcal{I}$.}
\BlankLine

\Fn{\FMain{$X$}}{
    $S \leftarrow \text{agents that agent $1$ $\mathsf{EF1}$-envies in } X$\;
    \While{$S \neq \emptyset$}{
        $i \leftarrow \underset{\ell \in [n] \sm \{1\}}{\arg\max} \set{v(X_{\ell} \sm \gb{X}{\ell})}$\;\label{line:choice-of-i-basic}
\BlankLine
\makebox[0.8\linewidth][l]{$\left.\begin{array}{l}
    X_i \leftarrow X_i \sm \gb{X}{i} \\
    X_1 \leftarrow X_1 \cup \gb{X}{i}
\end{array}
\right\} \: \text{Transfer } \gb{X}{i} \text{ from } i \text{ to }1$}\label{line:transfer-best-basic}
\BlankLine
    Update $S$
}
\KwRet{$X = (X_1, \ldots, X_n)$}\;
}
\end{algorithm}

\begin{lemma}\label{lem:identical-basic-correctness}
$\efconn$ is possible for identical monotone valuations, when the input allocation is near-$\mathsf{EF1}$.
\end{lemma}

\begin{proof}
 Let $X = (X_1,\ldots,X_n)$ be the input near-$\mathsf{EF1}$ allocation. We write $v$ to denote the monotone valuation that is common to all the agents. Let $i \coloneq \underset{\ell \in [n]\setminus\{1\}}{\arg\max} \set{v(X_{\ell} \sm \gb{X}{\ell})}$. That is, in the allocation $X$, agent $i$ is the one (other than agent $1$, of course) who would possess the most valuable bundle, even after losing her most valuable item. As indicated in \linkline{line:transfer-best-basic} of \Cref{algo:identical-basic}, we transfer $\gb{X}{i}$ from $i$ to agent $1$, and denote the resulting allocation by $Y$. We claim that this is a valid transfer, i.e., $Y$ is near-$\mathsf{EF1}$ and this transfer does not lead to an increase in the amount of $\mathsf{EF1}$-envy (\cref{def:amount-of-ef1-envy}) that agent $1$ has towards any other agent.

\begin{itemize}
    \item Observe that agent $i$ does not $\mathsf{EF1}$-envy any agent $j\neq 1$ after the transfer. We have $Y_1 = X_1 \cup \gb{X}{i}$, $Y_i = X_i \sm \gb{X}{i}$, and $Y_j = X_j$ for all $j \in [n] \sm \set{1,i}$. Then, we have, for each $j \notin \set{1,i}$,
\begin{align} \label{eqn:i_doesn't_ef1_envy_j-basic}
v(Y_i) &= v(X_i \sm \gb{X}{i}) \geq v(X_j \sm \gb{X}{j}) = v(Y_j \sm \gb{Y}{j}),
\end{align}

where the inequality in \eqref{eqn:i_doesn't_ef1_envy_j-basic} follows from the choice of $i$ in \linkline{line:choice-of-i-basic} of \Cref{algo:identical-basic}. Hence, agent $i$ does not $\mathsf{EF1}$-envy any agent $j \neq 1$ after the transfer.

\item Next, we prove that no agent $\mathsf{EF1}$-envies agent $1$ after this transfer (i.e., in $Y$). Note that the transfer was made because agent $1$ had $\mathsf{EF1}$-envy towards some agent, say agent $2$ in $X$. So, for each $j \notin \set{1,i}$, we have the following:

\begin{align} \label{eqn:j_doesn't_ef1_envy_1-basic}
v(Y_j)=v(X_j) \geq v(Y_i) &= v(X_i \setminus \gb{X}{i}) \tag{$X$ is near-$\mathsf{EF1}$}\\
& \geq v(X_2 \setminus \gb{X}{2}) \tag{by our choice of $i$}\\
&> v(X_1), \tag{agent $1$ $\mathsf{EF1}$-envies agent $2$}
\nonumber
\end{align}
and
\begin{align*}
    v(X_1) = v(Y_1 \setminus \gb{X}{i}) \geq v(Y_1 \setminus \gb{Y}{1}). \tag{since $X_1 \subsetneq Y_1$}
\end{align*}

Therefore, we have $v(Y_j) \geq v(Y_1 \setminus \gb{Y}{1})$ for all $j \neq 1$.

\item Observe that agent $1$ must have $\mathsf{EF1}$-envied agent $i$ in $X$. As assumed earlier, let agent $2$ be one of the agents who is $\mathsf{EF1}$-envied by agent $1$ in $X$. So, $v(X_2\setminus \gb{X}{2})>v(X_1)$. But, since agent $i$ was chosen for the transfer, we must have had $v(X_i \sm \gb{X}{i})\geq v(X_2 \sm \gb{X}{2}) > v(X_1)$.

\item Since $1$ has gained an item, the $\mathsf{EF1}$-envy agent $1$ has towards any other agent in $Y$ cannot be larger than that in $X$. 

\end{itemize}
Overall, we have shown that the transfer in \linkline{line:transfer-best-basic} is valid and results in an increase of one item in the bundle of agent $1$ (as long as agent $1$ is not $\mathsf{EF1}$-happy). Hence, \Cref{algo:identical-basic} eventually terminates with an $\mathsf{EF1}$ allocation.
\end{proof}


\subsubsection{\texorpdfstring{$\efconn$}{EF1-Restoration} for Arbitrary Allocations} \label{subsubsec:identical-general}

In this section, we prove \hyperref[thm:identical-general-restated]{\Cref*{thm:identical-general}}, which we restate here for convenience.

\thmidentical*
\label{thm:identical-general-restated}

The $\efconn$ algorithm \Cref{algo:identical-general} that we describe for an arbitrary input allocation is similar to \Cref{algo:identical-basic}. Instead of agent $1$, the recipient of an item at each iteration will simply be one of the poorest agents at the beginning of that iteration. The agent who transfers an item will still be one who would be least affected by giving away her most valuable item at the beginning of that iteration.

\begin{algorithm}[htbp]
\caption{\texorpdfstring{$\efconn$}{EF1-Restoration} for identical valuations} \label{algo:identical-general}
\SetAlgoLined
\SetKwInOut{Input}{Input}
\SetKwInOut{Output}{Output}
\SetKwFunction{FMain}{EF1\_Restorer\_Goods}
\SetKwProg{Fn}{}{:}{}
\SetKw{KwTo}{to}
\SetKw{KwRet}{return}
\DontPrintSemicolon

\Input{\text{ } A fair division instance $\mathcal{I}=\left<[n], \G, v \right>$ with identical monotone (positive) valuation $v$ and a starting allocation $X = (X_1, \ldots, X_n)$.}
\Output{\text{ } An $\mathsf{EF1}$ allocation in $\mathcal{I}$.}
\BlankLine

\Fn{\FMain{$X$}}{
    \While{$X$ is not $\mathsf{EF1}$}{
        $i \leftarrow \underset{\ell \in [n]}{\arg\max} \set{v(X_{\ell} \sm \gb{X}{\ell})}$\;\label{line:choice-of-i-general}
        $j \leftarrow \underset{\ell \in [n]}{\arg\min} \set{v(X_{\ell})}$\; \label{line:choice-of-j-general}
\BlankLine
\makebox[0.8\linewidth][l]{$\left.\begin{array}{l}
    Y_i \leftarrow X_i \sm \gb{X}{i} \\
    Y_j \leftarrow X_j \cup \gb{X}{i}
\end{array}
\right\} \: \text{Transfer } \gb{X}{i} \text{ from } i \text{ to }j$}\label{line:transfer-best-general}
\BlankLine
    Set $X\leftarrow Y$, where $Y_k=X_k, \ \forall k\notin\{i,j\}$
}
\KwRet{$X = (X_1, \ldots, X_n)$}\;
}
\end{algorithm}

We begin by making some simple, yet crucial, observations, that will finally lead to establishing the correctness of \Cref{algo:identical-general}, and bounding the number of transfers that the algorithm makes before it reaches an $\mathsf{EF1}$ allocation.

\begin{lemma} \label{lem:i-happy-j-ef1-envied-general}
Let $Z = (Z_1,\ldots,Z_n)$ be an arbitrary allocation at the beginning of any iteration of \Cref{algo:identical-general}. Let $i \coloneq \underset{\ell \in [n]}{\arg\max} \set{v(Z_{\ell} \sm \gb{Z}{\ell})}$, and let $j \coloneq \underset{\ell \in [n]}{\arg\min} \set{v(Z_{\ell})}$. Then, (i) no agent envies agent $j$, (ii) agent $i$ does not $\mathsf{EF1}$-envy any other agent, and (iii) if $Z$ is not an $\mathsf{EF1}$ allocation, then agent $j$ $\mathsf{EF1}$-envies agent $i$.
\end{lemma}

\begin{proof}
(i) This directly follows from our choice of $j$ (\linkline{line:choice-of-j-general}), as $j$ is the poorest agent in $Z$.

(ii) Let $\ell \in [n]$. Then, $v(Z_i) \geq v(Z_i \sm \gb{Z}{i}) \geq v(Z_{\ell} \sm \gb{Z}{\ell})$, by our choice of $i$ (\linkline{line:choice-of-i-general}). Hence, for all $\ell \in [n]$, agent $i$ does not $\mathsf{EF1}$-envy agent $\ell$.

(iii) Let us assume, to the contrary, that $Z$ is not an $\mathsf{EF1}$ allocation, and agent $j$ does not $\mathsf{EF1}$-envy agent $i$. Then, for all $\ell,\ell' \in [n]$, we have $v(Z_{\ell}) \geq v(Z_j) \geq v(Z_i \sm \gb{Z}{i}) \geq v(Z_{\ell'} \sm \gb{Z}{\ell'})$, where the first and last inequalities follow from our choices of $j$ and $i$ respectively. That is, for any two agents $\ell$ and $\ell'$, $\ell$ does not $\mathsf{EF1}$-envy $\ell'$ in $Z$. Thus $Z$ is an $\mathsf{EF1}$ allocation, leading to a contradiction. Hence $j$ $\mathsf{EF1}$-envies $i$ in $Z$.
\end{proof}

\begin{lemma}\label{lem:every-transfer-valid}
Each transfer made by \Cref{algo:identical-general} is a valid transfer. Furthermore, as long as we do not reach an $\mathsf{EF1}$ allocation, there exists a valid transfer and \Cref{algo:identical-general} executes one such transfer.
\end{lemma}

\begin{proof}
Let $Z = (Z_1,\ldots,Z_n)$ be the allocation at the beginning of any iteration. Let $i \ceq \underset{k \in [n]}{\arg\max} \set{v(Z_k \sm \gb{Z}{k})}$, and let $j \ceq \underset{k \in [n]}{\arg\min} \set{v(Z_k)}$. Denote the amount of $\mathsf{EF1}$-envy from $\ell$ to $\ell'$ in $Z$ by ${\varepsilon}_{\ell \to \ell'}^{Z} \coloneq v(Z_{\ell'} \sm \gb{Z}{\ell'}) - v(Z_{\ell})$. Thus, $\ell$ $\mathsf{EF1}$-envies $\ell'$ in $Z$ iff ${\varepsilon}_{\ell \to \ell'}^{Z} > 0$.

In each iteration, \Cref{algo:identical-general} transfers $\gb{Z}{i}$ from agent $i$ to agent $j$ (\linkline{line:transfer-best-general}). Let the resulting allocation be called $Y$. We first prove that this is a valid transfer, i.e., for any pair of agents $\ell, \ell' \in [n]$, 
${\varepsilon}_{\ell \to \ell'}^{Y} \leq {\varepsilon}_{\ell \to \ell'}^{Z}$. This is obvious for $\ell = j$ and/or $\ell' = i$, since $v(Y_i) < v(Z_i)$ and $v(Y_j) > v(Z_j)$. Also, it is obvious when $\ell,\ell'\in [n] \sm \{i,j\}$ as their bundles remain unchanged. Thus, we only need to show that no agent $\mathsf{EF1}$-envies $j$, and $i$ does not $\mathsf{EF1}$-envy anyone after the transfer.

\begin{itemize}
\item \emph{$\mathsf{EF1}$-envy from $\ell \; (\neq i)$ towards $j$:} 
\begin{equation} \label{eqn:l-to-j-ef1-envy-general}
{\varepsilon}_{\ell \to j}^{Y} = v(Y_{j} \sm \gb{Y}{j}) - v(Y_{\ell}) \leq v(Y_j \sm \gb{Z}{i}) - v(Z_{\ell}) = v(Z_j) - v(Z_{\ell}) \leq 0.    
\end{equation}
Therefore, no agent $\ell \; (\neq i)$ $\mathsf{EF1}$-envies $j$ after the transfer of $\gb{Z}{i}$ from $i$ to $j$.

\item \emph{$\mathsf{EF1}$-envy from $i$ to $\ell' \; (\neq j)$:} From the choice of $i$ (\linkline{line:choice-of-i-general}), we get
\begin{equation} \label{eqn:i-to-l'-ef1-envy-general}
{\varepsilon}_{i \to \ell'}^{Y} = v(Y_{\ell'} \sm \gb{Y}{\ell'}) - v(Y_{i}) = v(Z_{\ell'} \sm \gb{Z}{\ell'}) - v(Y_i) = v(Z_{\ell'}\sm \gb{Z}) - v(Z_i \sm \gb{Z}{i}) \leq 0.    
\end{equation}

\item \emph{$\mathsf{EF1}$-envy from $i$ to $j$:}
\begin{equation} \label{eqn:i-to-j-ef1-envy-general}
{\varepsilon}_{i \to j}^{Y} = v(Y_{j} \sm \gb{Y}{j}) - v(Y_{i}) \leq v(Y_j \sm \gb{Z}{i}) - v(Y_i) = v(Z_j) - v(Z_i \sm \gb{Z}{i}) \leq 0.    
\end{equation}
Therefore, $i$ does not $\mathsf{EF1}$-envy $j$ after transferring $\gb{Z}{i}$ to $j$.
\end{itemize}

This shows that \Cref{algo:identical-general} does not introduce any new $\mathsf{EF1}$-envy between any pair of agents after any transfer. Thus, each transfer is a valid transfer.

By part (iii) of \cref{lem:i-happy-j-ef1-envied-general}, there always exists a valid transfer that can be made as long as the allocation is not $\mathsf{EF1}$ -- namely, the transfer that \Cref{algo:identical-general} makes at each iteration. This completes the proof.
\end{proof}

Using \cref{lem:every-transfer-valid}, we prove that \Cref{algo:identical-general} terminates with an $\mathsf{EF1}$ allocation after making at most $m$ (valid) transfers.

\begin{corollary} \label{cor:termination-of-identical-general}
    Every item is transferred at most once during the execution of \Cref{algo:identical-general}. Hence, \Cref{algo:identical-general} terminates with an $\mathsf{EF1}$ allocation after making at most $m$ (valid) transfers.
\end{corollary}

\begin{proof}
    Consider the allocation $X^{(r)}$ at the beginning of iteration $r$ in the execution of \Cref{algo:identical-general}.
    Let $\mathcal{U}^{(r)} \subset [n]$ denote the set of $\mathsf{EF1}$-unhappy agents, $\mathcal{H}^{(r)} \coloneq [n] \setminus \mathcal{U}^{(r)}$ the set of $\mathsf{EF1}$-happy agents, and $\mathcal{E}^{(r)}$ the set of $\mathsf{EF1}$-envied agents, in $X^{(r)}$. From \cref{lem:every-transfer-valid}, the sets of $\mathsf{EF1}$-unhappy and $\mathsf{EF1}$-envied agents never grow, i.e., $\mathcal{U}^{(r)} \subseteq \mathcal{U}^{(s)}$, and $\mathcal{E}^{(r)} \subseteq \mathcal{E}^{(s)}$, whenever $r < s$. Moreover, by \cref{lem:i-happy-j-ef1-envied-general}, a transfer always takes place from an agent in $\mathcal{H}^{(r)} \cap \mathcal{E}^{(r)}$ to an agent in $\mathcal{U}^{(r)} \setminus \mathcal{E}^{(r)}$ at iteration $r$. Thus, no agent who receives an item during the course of \Cref{algo:identical-general} ever transfers an item, and vice versa. Thus, each item is transferred at most once over the course of the execution of \Cref{algo:identical-general}. This completes the proof.
\end{proof}

\cref{lem:every-transfer-valid} and \cref{cor:termination-of-identical-general} together establish \hyperref[thm:identical-general-restated]{\Cref*{thm:identical-general}}.


\subsubsection{ \texorpdfstring{$\efconn$}{EF1-Restoration} for Chores} \label{subsec:identical-chores-general}

In this section, we give an algorithm (\Cref{algo:identical-chores-general}) for the $\efconn$ problem for identical monotone valuations, in the case where the items to be allocated are all \emph{chores}, i.e., every item has negative marginal value for each agent. This algorithm is analogous to \Cref{algo:identical-general}, and the only changes one must make in this setting is to modify the definition of $\mathsf{EF1}$-envy, which we shall now see.

Formally, consider the fair division instance $\tilde{\mathcal{I}} \ceq \left<[n], \mathcal{C}, \set{v_i \colon 2^{\mathcal{C}} \to {\mathbb{R}}_{<0}}_{i \in [n]}\right>$, where $\mathcal{C}$ is the set of \emph{chores}, each of which has a negative marginal value for every agent (i.e., for all $i \in [n]$, $S \subseteq \mathcal{C}$, $c \in \mathcal{C} \sm S$, we have $v_i(S \cup c) - v_i(S) < 0$), and are to be fairly allocated among the agents $[n]$. Let $m \coloneq \abs{\mathcal{C}}$. Then, we have the following definitions. Envy is defined exactly the same way as in the case of goods (see \cref{def:ef1-envy}).

\begin{definition}[$\mathsf{EF1}$-envy in the case of chores]\label{def:ef1-envy-chores}
An agent $i$ is said to \emph{$\mathsf{EF1}$-envy} an agent $j$ if, for all $c \in X_i$, we have $v_i(X_i \sm c) < v_i(X_j)$. That is, $i$ continues to envy $j$ even after virtually \emph{removing any one chore} from her own bundle. 
We say $i$ is \emph{$\mathsf{EF1}$-happy} in $X$ if she does not $\mathsf{EF1}$-envy any other agent in $X$.
\end{definition}

We may analogously modify the definition of the amount of $\mathsf{EF1}$-envy felt by one agent for another.

\begin{definition}[Amount of $\mathsf{EF1}$-envy in the case of chores]\label{def:amount-of-ef1-envy-chores}
Let $X = (X_1,\ldots,X_n)$ be an allocation of chores. The amount of $\mathsf{EF1}$-envy that agent $i$ has towards agent $j$ in $X$, denoted by $\e{i\to j}(X)$, is defined as $\e{i\to j}(X) \ceq \underset{c \in X_i}{\min} \set{v_i(X_j)-v_i(X_i \sm c)}$.\end{definition}

With these modified definitions, we may define an $\mathsf{EF1}$ allocation exactly as we did in \cref{def:ef1}.

\begin{definition}[Envy-freeness up to one item ($\mathsf{EF1}$) for chores]\label{def:ef1-chores}
An allocation $X$ of chores is said to be $\mathsf{EF1}$ if every agent is $\mathsf{EF1}$-happy in $X$.
\end{definition}

Note that $X$ is an $\mathsf{EF1}$ allocation if and only if $\e{i \to j}(X)\leq 0$ for each $i,j\in [n]$, as in the case of goods.

The goal of this section is to prove the following theorem.

\begin{theorem} \label{thm:identical-chores-general}
    Given a fair division instance on chores with an identical monotone valuation $v \colon 2^{\mathcal{C}} \to \mathbb{R}_{<0}$ over chores and an input allocation $X$, $\efconn$ is possible.
\end{theorem}

At each iteration of our $\efconn$ algorithm \Cref{algo:identical-chores-general} for chores with an identical monotone valuation $v$, we pick an agent $i$ who would possess the least valuable bundle even after losing her least liked chore, and ease her burden by transferring such a chore $\cw{X}{i}$ to any of the richest agents $j$. Formally, $\cw{X}{i}$ is a chore in agent $i$'s bundle $X_i$ that has the most negative marginal value for all the agents, i.e., $\cw{X}{i} \ceq \underset{c \in X_i}{\arg\max} \set{v(X_i \sm c)}$.

In \cref{lem:identical-chores-general-correctness}, we prove that such a transfer is valid, and makes progress towards reaching an $\mathsf{EF1}$ allocation.

\begin{algorithm}[htbp]
\caption{\texorpdfstring{$\efconn$}{EF1-Restoration} for chores with identical valuations} \label{algo:identical-chores-general}
\SetAlgoLined
\SetKwInOut{Input}{Input}
\SetKwInOut{Output}{Output}
\SetKwFunction{FMain}{EF1\_Restorer\_Chores}
\SetKwProg{Fn}{}{:}{}
\SetKw{KwTo}{to}
\SetKw{KwRet}{return}
\DontPrintSemicolon

\Input{\text{ }A fair division instance $\tilde{\mathcal{I}} = \left<[n], \mathcal{C}, v \right>$ with an identical monotone valuation $v \colon 2^{\mathcal{C}} \to \mathbb{R}_{<0}$, and a starting allocation $X = (X_1, \ldots, X_n)$ of the set $\mathcal{C}$ of chores.}
\Output{\text{ }An $\mathsf{EF1}$ allocation in $\mathcal{I}$.}
\BlankLine

\Fn{\FMain{$X$}}{
    \While{$X$ is not $\mathsf{EF1}$}{
        $i \leftarrow \underset{\ell \in [n]}{\arg\min} \set{v(X_{\ell} \sm \cw{X}{\ell})}$\;\label{line:choice-of-i-chore}
        $j \leftarrow \underset{\ell \in [n]}{\arg\max} \set{v(X_{\ell})}$\; \label{line:choice-of-j-chore}
\BlankLine
\makebox[0.8\linewidth][l]{$\left.\begin{array}{l}
    Y_i \leftarrow X_i \sm \cw{X}{i} \\
    Y_j \leftarrow X_j \cup \cw{X}{i}
\end{array}
\right\} \: \text{Transfer } \cw{X}{i} \text{ from } i \text{ to }j$}\label{line:transfer-chore}
\BlankLine
    Set $X \leftarrow Y$, where $Y_k=X_k, \ \forall \ k\notin\{i,j\}$
}
\KwRet{$X = (X_1, \ldots, X_n)$}\;
}
\end{algorithm}

\begin{lemma} \label{lem:j-happy-i-ef1-envies-chores-general}
Let $X = (X_1,\ldots,X_n)$ be an arbitrary allocation of chores. Let $i \coloneq \underset{\ell \in [n]}{\arg\min} \set{v(X_{\ell} \sm \cw{X}{\ell})}$, and let $j \coloneq \underset{\ell \in [n]}{\arg\max} \set{v(X_{\ell})}$. Then, (i) agent $j$ does not envy any other agent, (ii) agent $i$ is not $\mathsf{EF1}$-envied by any other agent, and (iii) if $X$ is not an $\mathsf{EF1}$ allocation, then agent $i$ $\mathsf{EF1}$-envies agent $j$.
\end{lemma}

\begin{proof}
The proof is analogous to the proof of \cref{lem:i-happy-j-ef1-envied-general} and is omitted.
\end{proof}

\begin{lemma}\label{lem:identical-chores-general-correctness}
$\efconn$ for chores is possible, for identical monotone valuations.
\end{lemma}

\begin{proof}
 We claim that \Cref{algo:identical-chores-general} is an $\efconn$ algorithm for chores with identical monotone valuations. Let $X = (X_1,\ldots,X_n)$ be the input allocation, with a negative valuation $v$ that is common to all the agents. Let $i \coloneq \underset{\ell \in [n]}{\arg\min} \set{v(X_{\ell} \sm \cw{X}{\ell})}$, and let $j \coloneq \underset{\ell \in [n]}{\arg\max} \set{v(X_{\ell})}$. As indicated in \linkline{line:transfer-chore} of \Cref{algo:identical-chores-general}, we transfer $\cw{X}{i}$ from agent $i$ to agent $j$, and denote the resulting allocation by $Y$. We claim that this is a valid transfer. To this end, it is enough to verify that no agent $\mathsf{EF1}$-envies $i$, and agent $j$ does not $\mathsf{EF1}$-envy anybody after this transfer. We have

\begin{itemize}

\item \emph{$\mathsf{EF1}$-envy from $\ell \in [n] \sm \set{i,j}$ to $i$:} We have
\begin{equation} \label{eqn:l-to-i-ef1-envy-chores}
{\varepsilon}_{\ell \to i}^{Y} =  v(Y_{i}) - v(Y_{\ell} \sm \cw{Y}{\ell})= v(X_i \sm \cw{X}{i}) - v(X_{\ell} \sm \cw{X}{\ell}) \leq  0.
\end{equation}

The last inequality follows from our choice of $i$.
Therefore, no agent $\ell \in [n] \sm \set{i,j}$ $\mathsf{EF1}$-envies $i$ after the transfer of $\cw{X}{i}$ from $i$ to $j$.

\item \emph{$\mathsf{EF1}$-envy from $j$ to $i$:} We have
\begin{align}
{\varepsilon}_{j \to i}^{Y} = v(Y_{i}) - v(Y_{j} \sm \cw{Y}{j}) &= v(X_i \sm \cw{X}{i})-v((X_j \cup \cw{X}{i}) \sm \cw{Y}{j}) \notag
\\&\leq v(X_i) - v(X_j \sm \cw{Y}{j}) \notag
\\&\leq v(X_i)-v(X_j \sm \cw{X}{j})   = {\varepsilon}_{j \to i}^{X} \leq 0. \label{eqn:j-to-i-ef1-envy-chores}
\end{align}

Therefore, agent $j$ does not $\mathsf{EF1}$-envy agent $i$, even after receiving $\cw{X}{i}$ from agent $i$.

\item \emph{$\mathsf{EF1}$-envy from $j$ to $\ell \in [n] \sm \set{i,j}$}: We have
\begin{align}
{\varepsilon}_{j \to \ell}^{Y} = v(Y_{\ell})-v(Y_{j} \sm \cw{Y}{j}) &= v(X_{\ell})-v((X_j \cup \cw{X}{i}) \sm \cw{Y}{j})   \notag
\\&\leq v(X_{\ell})-v(X_j \sm \cw{Y}{j})  \notag
\\&\leq v(X_{\ell})-v(X_j \sm \cw{X}{j})  = {\varepsilon}_{j \to \ell}^{X} \leq 0, \label{eqn:j-to-l-ef1-envy-chores}
\end{align}

by part (i) of \cref{lem:j-happy-i-ef1-envies-chores-general}.Therefore, agent $j$ does not $\mathsf{EF1}$-envy any agent $\ell \neq i$, even after receiving $\cw{X}{i}$ from agent $i$.
\end{itemize}

This shows that \Cref{algo:identical-chores-general} does not introduce any new $\mathsf{EF1}$-envy among any pair of agents, after each transfer that it initiates. That is, every transfer that it initiates is \emph{valid}.

The proof that \Cref{algo:identical-chores-general} terminates with an $\mathsf{EF1}$ allocation follows from \cref{lem:j-happy-i-ef1-envies-chores-general} and is analogous to \cref{lem:every-transfer-valid} and \cref{cor:termination-of-identical-general}.
\end{proof}

It is straightforward to see that all the algorithms described in this section (\cref{subsec:ef1-restoration-identical-monotone}) run in polynomial time.


\subsection{\texorpdfstring{$\efopt$}{Optimal-EF1-Restoration} for Identical Additive Valuations} \label{subsec:NP-hard-opt-identical-additive}

When agents have identical valuations over goods, \cref{algo:identical-general} computes, in polynomial time, a sequence of \emph{valid transfers} that transforms an arbitrary input allocation into an $\mathsf{EF1}$ allocation. Moreover, $\mathsf{EF1}$ is restored after at most $m$ valid transfers, since each item is transferred at most once during the course of \cref{algo:identical-general} (see \cref{lem:every-transfer-valid} and \cref{cor:termination-of-identical-general}).

This naturally leads to an optimization question: given an arbitrary allocation, what is the minimum number of \emph{valid} transfers required to reach an $\mathsf{EF1}$ allocation? We refer to this problem as $\efopt$, which we restate here for convenience.

\ktransfers*
\label{prob:k-valid-transfers-EF1-restated}

The following result shows that $\efopt$ is strongly $\mathrm{NP}$-complete even for identical additive valuations.

\thmidenticaloptimal*
\label{thm:k-valid-transfers-EF1-NP-hard-identical-additive-restated}

\begin{proof}
We refer to the decision problem here: whether $\efconn$ is possible in at most $k$ valid transfers.
It is easy to see that this decision problem belongs to $\mathrm{NP}$, since, given a sequence of at most $k$ transfers, one can verify in polynomial time that every transfer is valid and that the resulting allocation is $\mathsf{EF1}$. Note that $k\leq m$ from \cref{cor:termination-of-identical-general}.
To show $\mathrm{NP}$-hardness, we give a polynomial-time reduction from the $3$-\textsc{Partition} problem defined below.

\begin{definition}[$3$-\textsc{Partition}]
Given positive integers $r$ and $B$, and a set of $3r$ positive integers $S = \{a_1,\ldots,a_{3r}\}$ such that $\sum_{i=1}^{3r} a_i = rB$, determine if $S$ can be partitioned into $r$ subsets $S_1,\ldots,S_r$ such that, for each $i \in [r]$, $\abs{S_i} = 3$, and $\sum_{x \in S_i} x = B$.
\end{definition}

$3$-\textsc{Partition} is strongly $\mathrm{NP}$-complete and remains so even if $B/4 < a_i < B/2$ for all $i \in [3r]$ (see \cite{GareyJ79}, Theorem 4.4, Page 99). 

\noindent{\bf Construction of the instance of $\efopt$: } Consider an instance $\mathcal{S} \coloneq (a_1,\ldots,a_{3r};B)$ of $3$-\textsc{Partition}, where $\sum_{i=1}^{3r} a_i = rB$, and $B/4 < a_i < B/2$ for all $i \in [3r]$. 
Let $\lambda \coloneq 3r^2+4r+1 \text{ and } T \coloneq \lambda B$. We construct an instance $\mathcal{I(S)}$ of $\efopt$ as follows.

\begin{itemize}
\item Set of agents $=\{s, 1,\ldots,r\}$. 
\item Set of items $\mathcal{M} = G \cup H$, where $G = \{g_i \mid i \in [3r]\}$, and $H = \{h_1,\ldots,h_T\}$. 
\item Valuation (identical additive): $v(g_i) \coloneq \lambda a_i$, for each $i\in [3r]$, $v(h_i) \coloneq 1$ for each $i \in [T]$. Here $a_i$, $i\in [3r]$ are integers from the $3$-\textsc{Partition} instance.
\item Initial allocation: $X = (\mathcal{M},\emptyset,\ldots,\emptyset)$, where $X_s = \mathcal{M}$.
\item Valid transfer budget is $k \coloneq 3r$.
\end{itemize}

Next, we show that $\mathcal{S}$ is a \textsc{Yes}-instance of $3$-\textsc{Partition} if and only if $\mathcal{I(S)}$ is a \textsc{Yes}-instance of $\efopt$.\\

\noindent{\bf $3$-\textsc{Partition} $\implies$ $\efopt$:}

First, suppose $\mathcal{S} \coloneq (a_1,\ldots,a_{3r};B)$ is a \textsc{Yes}-instance of $3$-\textsc{Partition}. Then, there exists a partition $\{S_1,\ldots,S_r\}$ of $[3r]$ such that $\abs{S_j} = 3$ and $\sum_{i \in S_j} a_i = B$ for every $j \in [r]$. Consider the corresponding allocation $X' = (H,X'_1,\ldots,X'_r)$ in the instance $\mathcal{I(S)}$, where $X'_i=\{g_j \mid j\in S_i\}$ for $1 \leq i \leq 3r$. We prove that $X'$ is $\mathsf{EF1}$ and can be reached from the initial allocation $X$ by performing $k$-many valid transfers.

\begin{claim}\label{clm:valid-k}
The allocation $X'$ is $\mathsf{EF1}$, and can be reached from $X$ by a sequence of $k$ valid transfers.
\end{claim}
\begin{proof}
The allocation $X'$ is $\mathsf{EF1}$, and in fact $\mathsf{EF}$, since $v(X'_i)=v(X'_s)=T$ for each $i \in [r]$.

To begin with, recall that in the initial allocation $X$, we have $X_s = \mathcal{M}$ and $X_i = \emptyset$ for all $i \in [r]$. We construct a sequence of $3r = k$ valid transfers that results in the allocation $X'$; all these transfers are from agent $s$.

The first $r$ transfers goes as follows: 
for every $i \in [r]$, choose any one item listed in $X'_i$ and transfer it from agent $s$ to agent $i$.
It is trivial to see that these $r$ transfers are valid since for each agent in $1,\ldots,r$, their current bundle consists of a single item, so they cannot be $\mathsf{EF1}$-envied by anyone.

For the subsequent $2r$ transfers, until $X'$ is reached, we repeatedly choose an agent $i \in [r]$ (breaking ties arbitrarily) that has a minimum-valued bundle,
and transfer the next item as listed in $X'_i$ to her (from agent $s$).
Note that any agent $i \in [r]$ who has already received all three items listed in $X'_i$ has a utility of exactly $T$, whereas any agent $i \in [r]$ who is yet to receive $X'_i$ completely has a utility strictly less than $T$. This is because each $a_i < \frac{B}{2}$. So the agent $i$ with a minimum-valued bundle is among those who have not yet received their respective bundle $X'_i$.

Each such transfer is valid, as the removal of the most recently added item from the bundle of agent $i$ removes any envy towards agent $i$, and moreover, agent $s$ never has a bundle of value less than $T$, and hence does not envy anyone.
\end{proof}

\noindent{\bf $\efopt$ $\implies$ $3$-\textsc{Partition}:}

Conversely, suppose that $\mathcal{I(S)}$ is a \textsc{Yes}-instance of $\efopt$. Then, there is an $\mathsf{EF1}$ allocation $X'$ that can be obtained from the initial allocation $X$ using at most $k = 3r$ valid transfers. We show in \cref{clm:identical-converse} below that $X'$ can only be obtained by transferring all the items from $G$ and none of the items from $H$ from agent $s$ to others.

\begin{claim}\label{clm:identical-converse}
In $X'$, $X'_s=H$.
\end{claim}
\begin{proof}
For the sake of contradiction, let $X'_s\cap G\neq \emptyset$. Clearly, $X'_s=G$ cannot be true in an $\mathsf{EF1}$ allocation. Let us write $R=v(X'_s\cap G)>0$ and $u = v(H\cap X'_s)$. Thus, $v(X'_s)=R+u$ and we have the following:

\begin{equation}
    v(X'_1\cup\ldots\cup X'_r)=r\lambda B-R+T-u=rT-R+T-u\label{eqn:others}
\end{equation}
Note that, $v(X'_s\setminus \{x\})\geq u-1$ where $x$ is the largest valued item in $X'_s$. Hence, for $X'$ to be $\mathsf{EF1}$, we must have $v(X'_j)\geq u-1$ for each $j\in [r]$. Therefore,
\begin{equation}\label{eqn:lower-bound-others}
    v(X'_1 \cup \ldots \cup X'_r) \geq r(u-1)
\end{equation}
From equations~(\ref{eqn:others}) and (\ref{eqn:lower-bound-others}), we get 
\begin{eqnarray*}
    rT-R+T-u &\geq &r(u-1)\\
    (r+1)(T-u)+r& \geq & R
\end{eqnarray*}
Since $T-u\leq 3r$, we get $R\leq 3r^2+4r<\lambda$. However, $R$ must be a multiple of $\lambda$ since the value of each item in $G$ is a multiple of $\lambda$. Therefore, $R=0$, and thus agent $s$ must transfer all the items in $G$ to the others. Due to the transfer budget limit $k = 3r$, no item from $H$ can, therefore, be transferred. Thus $X'_s = H$.
\end{proof}
For $X'$ to be $\mathsf{EF1}$, $v(X'_i) \geq T-1$ for all $i\in [r]$. However, $v(X'_i)$ is a multiple of $\lambda$. Since $T = \lambda B$, we have $v(X'_i)\geq T$. Since $v(G) = rT$, for each agent $i \in [r]$, $v(X'_i) = T$ must hold. The corresponding sets of items in the partition instance then form a solution to $3$-\textsc{Partition}: undoing the scaling by $\lambda$, the items allocated to each agent $i \in [r]$ correspond to a subset of $\{a_1,\ldots,a_{3r}\}$ summing to $B$. Finally, since each $a_i$ lies strictly between $B/4$ and $B/2$, each such subset must contain exactly three elements. Hence, $(a_1,\ldots,a_{3r};B)$ is a \textsc{Yes}-instance of the $3$-\textsc{Partition} instance $\mathcal{S}$. This completes our proof.
\end{proof}

\begin{remark}
As an immediate consequence of \hyperref[thm:k-valid-transfers-EF1-NP-hard-identical-additive-restated]{\Cref*{thm:k-valid-transfers-EF1-NP-hard-identical-additive}}, the optimization problem of computing the minimum number of valid transfers required to reach an $\mathsf{EF1}$ allocation is strongly $\mathrm{NP}$-complete, even when  agents have identical additive valuations.
\end{remark}

\begin{remark} \label{rem:opt-ef1-chores-np-hard}
    The analog of $\efopt$ for the case in which all items are \emph{chores} and agents have identical additive valuations is also strongly $\mathrm{NP}$-complete by following the same techniques used to prove \hyperref[thm:k-valid-transfers-EF1-NP-hard-identical-additive-restated]{\Cref*{thm:k-valid-transfers-EF1-NP-hard-identical-additive}}. In the above reduction, negating the values in $\mathcal{I(S)}$ gives a reduction for $\mathrm{NP}$-hardness of $\efopt$ on chores. This is because, in the goods reduction, if an agent $i$ $\mathsf{EF1}$-envies agent $j$, i.e., $v(X_i)<v(X_j \sm \gb{X}{j})$, then negating the values makes the item $\gb{X}{j}$ the same as $\cw{X}{j}$ and the values of the bundles of $i$ and $j$ satisfy $v(X_i)>v(X_j\sm \cw{X}{j})$. Thus $j$ $\mathsf{EF1}$-envies $i$ in the chores instance. The same sequence of transfers then restores $\mathsf{EF1}$.
\end{remark}

\subsection{\texorpdfstring{$\efconn$}{EF1-Restoration} on Mixed Manna} \label{subsec:ef1-mixed-case}

In light of the positive results for goods and chores, it is a natural question to extend them to the setting of mixed manna. In this section, we discuss $\efconn$ for the \emph{mixed manna} setting (where the set of items to be allocated is a mix of goods and chores), and show that it may not always be possible, even in the case of identical additive valuations.

Formally, we consider the $\efconn$ problem with an identical additive valuation $v \colon 2^{\G} \to \mathbb{R}$, where $v$ may take on positive as well as negative values, i.e., the set of items contains both \emph{goods and chores}. We refer to this problem as the $\efconn$ problem for \emph{mixed manna}. Note that, $\mathsf{EF1}$ allocations are guaranteed to always exist for the mixed manna setting as well \cite{AzizCIW19}.  We begin by formally defining what would an $\ef1$ allocation mean for mixed manna. 

\begin{definition}[$\mathsf{EF1}$-envy in the case of mixed manna]\label{def:ef1-envy-mixed}
We say an agent $i$ \emph{$\mathsf{EF1}$-envies}  agent $j$ in an allocation $X$ of goods and chores if, 
\begin{enumerate}[label=(\roman*), ref=\roman*]
    \item for all chores $c \in X_i$, we have $v_i(X_i \sm c) < v_i(X_j)$, and
    \item for all goods $g \in X_j$, we have $v_i(X_i) < v_i(X_j \sm g)$.
\end{enumerate}

That is, $i$ continues to envy $j$ even after virtually \emph{removing any one chore} from her own bundle, or after \emph{removing any one good} from $j$'s bundle. 
We say $i$ is \emph{$\mathsf{EF1}$-happy} in $X$ if she does not $\mathsf{EF1}$-envy any other agent.
\end{definition}

With this modified definition, we may define an $\mathsf{EF1}$ allocation and a near-$\mathsf{EF1}$ allocation exactly as we did in \cref{def:ef1} and \cref{def:near-ef1}.

\begin{definition}[Envy-freeness up to one item ($\mathsf{EF1}$) for mixed manna]\label{def:ef1-mixed}
An allocation $X$ of a set of goods and chores is said to be $\mathsf{EF1}$ if every agent is $\mathsf{EF1}$-happy in $X$.
\end{definition}

\begin{definition}[Near-$\mathsf{EF1}$ allocations of mixed manna]\label{def:near-ef1-mixed}
An allocation $X$ of a set of goods and chores is said to be near-$\mathsf{EF1}$ if exactly one \emph{fixed} agent is not $\mathsf{EF1}$-happy in $X$. As usual, we let agent $1$ denote the unhappy agent, without loss of generality.
\end{definition}

\begin{example} \label{example:identical-ef1-mixed-counterexample}

We present a simple $\efconn$ instance (\cref{tab:identical-ef1-mixed-counterexample}) with $n = 4$ identical agents with additive valuations over $m = 17$ items, and an input near-$\mathsf{EF1}$ allocation (\cref{def:near-ef1}), from where  no valid transfer or valid exchange of items exist.

\begin{table}[htbp]
\centering
\setlength{\arraycolsep}{5pt}
\[
\begin{array}{r|ccccccccccccccccc}
  & g_1 & g_2 & g_3 & g_4 & g_5 & g_6 & g_7 & g_8 & g_9 & g_{10} & g_{11} & g_{12} & g_{13} & g_{14} & g_{15} & g_{16} & g_{17} \\
\hline
\text{$v(g)$} & \textcolor{violet}{3} & \textcolor{violet}{-10} & \textcolor{teal}{10} & \textcolor{teal}{-1} & \textcolor{teal}{-1} & \textcolor{teal}{-1} & \textcolor{teal}{-1} & \textcolor{teal}{-1} & \textcolor{teal}{-1} & \textcolor{magenta}{6} & \textcolor{magenta}{-1} & \textcolor{magenta}{-1} & \textcolor{magenta}{-1} & \textcolor{magenta}{-1} & \textcolor{magenta}{-1} & \textcolor{cyan}{-2} & \textcolor{cyan}{-1}
\end{array}
\]

\vspace{0.25ex}

\[
\setlength{\arraycolsep}{6pt}
\begin{array}{@{}r|l@{}}
\text{Agent} & \text{Items in } X_i \\
\hline
\textcolor{violet}{1} & \{g_1,\, g_2\} \\
\textcolor{teal}{2} & \{g_3,\, g_4,\, g_5,\, g_6,\, g_7,\, g_8,\, g_9\} \\
\textcolor{magenta}{3} & \{g_{10},\, g_{11},\, g_{12},\, g_{13},\, g_{14},\, g_{15}\} \\
\textcolor{cyan}{4} & \{g_{16},\, g_{17}\}
\end{array}
\]

\caption{A near-$\mathsf{EF1}$ allocation $X$ for which no valid transfer or exchange exists}
\label{tab:identical-ef1-mixed-counterexample}
\end{table}

Observe that the only $\mathsf{EF1}$-envy in $X$ is from agent $1$ towards agent $2$, so $X$ is near-$\mathsf{EF1}$. We leave it to the readers to verify that every transfer or exchange of items leads to new $\mathsf{EF1}$-envy being created (i.e., removal of a chore from own bundle or removal of a good from the envied bundle is not sufficient to resolve envy), rendering all such operations invalid. Hence, $\efconn$ may not be possible in the mixed manna setting, even for identical additive valuations. \hfill $\qed$
\end{example}



\section{\texorpdfstring{$\efconn$}{EF1-Restoration} for Additive Binary Valuations} \label{sec:additive-binary}

In this section, we discuss $\efconn$ for additive binary valuations, wherein each item is valued either $0$ or $1$ by every agent, and the value of any bundle is the sum of values of its constituent items. In contrast to the case of identical monotone valuations (\cref{algo:identical-general}), an $\mathsf{EF1}$ allocation need not be reachable from an arbitrary input allocation via valid transfers for additive binary valuations. Indeed, in \cref{subsec:additive-binary-NP-hard}, we prove $\mathrm{NP}$-hardness of the budgeted variant $\efopt$, and also prove that the same reduction establishes $\mathrm{NP}$-hardness of the $\efconn$ problem for additive binary valuations (\Cref{thm:EF1-Restoration-NP-hard-additive-binary}).

Next, in \cref{subsec:graphical}, we identify an interesting subclass\footnote{Graphical valuations in fair division context is introduced by Christodoulou et al. \cite{christodoulou2023fair}} of additive binary valuations for which $\efconn$ is tractable -- namely, the class of \emph{multi-graphical valuations}, in which every item is valued (at $1$) by at most two agents. For this setting, we design an efficient algorithm (\Cref{algo:graphical-binary}) that finds a positive solution to $\efconn$ using an \emph{optimal} number of valid transfers (see \cref{thm:binary-arb}).

\subsection{\texorpdfstring{$\mathrm{NP}$}{NP}-hardness of \texorpdfstring{$\efconn$}{EF1-Restoration} and \texorpdfstring{$\efopt$}{Optimal-EF1-Restoration}} \label{subsec:additive-binary-NP-hard}

In this section, we prove \hyperref[thm:EF1-Restoration-NP-hard-additive-binary-restated]{\Cref*{thm:EF1-Restoration-NP-hard-additive-binary}}, which we restate here for convenience.

\theorembinary*
\label{thm:EF1-Restoration-NP-hard-additive-binary-restated}

\begin{proof}

To prove $\mathrm{NP}$-hardness of $\efopt$, we present a reduction from the \textsc{Independent Set} problem, defined below.

\begin{definition}[\textsc{Independent Set}]
Given an undirected graph $G = (V,E)$ and a positive integer $k \leq \abs{V}$, the goal is to decide whether $G$ contains an independent set of size at least $k$, that is, a set $S \subseteq V$ with $\abs{S} \geq k$ such that no two vertices in $S$ are adjacent.
\end{definition}

Let $\mathcal{S} \coloneq (G=(V,E),\;k)$ be an instance of \textsc{Independent Set}. We construct an instance $\mathcal{I(S)}$ of $\efopt$ as follows.\\

\noindent{\bf The reduction to $\efopt$:}
\begin{itemize}
\item The set of agents $\mathcal{N}=\{p_v \mid v \in V\} \cup \{c,s\} \cup \{a_e \mid e \in E\} \cup\{w_j,r_j \mid j \in [k+1]\}$. 
\item The set of items $\mathcal{M}=\{g_v,\ell_v \mid v \in V\} \cup \{\ell_c\} \cup \{d_1,\ldots,d_{k+1}\} \cup \{\ell_e \mid e \in E\} \cup \{\ell_j,\ell'_j \mid j\in [k+1]\}$.
\item Valid transfer budget $\coloneq k$.
\end{itemize}
The valuation and initial allocation $X$ are described informally below, and summarized in \cref{tab:binary-valuation-matrix}.

\begin{itemize}
\item For each vertex $v \in V$, we have a {\em vertex item} $g_v$, a {\em label item} $\ell_v$, and a {\em gadget agent} $p_v$ with $X_{p_v} \coloneq \{g_v,\ell_v\}$.

\item We have a {\em color agent} $c$ and a {\em label item} $\ell_c$, where $X_c \coloneq \{\ell_c\}$,

\item An agent $s$ and {\em dummy items} $d_1,\ldots,d_{k+1}$, where $X_s\coloneq\{d_1,\ldots,d_{k+1}\}$,

\item For each edge $e \in E$, a label item $\ell_{a_e}$ and an edge agent $a_e$ with 
$X_{a_e} \coloneq \{\ell_{a_e}\}$.

\item For each $j \in [k+1]$, two {\em watcher agents} $w_j$ and $r_j$, with label items $\ell_j$ and $\ell'_j$. Their initial bundles are $X_{w_j} \coloneq \{\ell_j\}$ and $X_{r_j} \coloneq \{\ell'_j\}$.
\end{itemize}

All valuations are additive and binary, defined as follows. When we say an agent $a$ \enquote{values} an item $g$, we mean $v_a(g)=1$.

\begin{itemize}
    \item each $p_v$ values no item;
    \item $c$ values every vertex item $g_v \in V$ and every dummy item $d_1,\ldots,d_{k+1}$;
    \item $s$ values only the dummy items $d_1,\ldots,d_{k+1}$;
    \item for each edge $e = (u,v) \in E$, the agent $a_e$ values only $g_u$ and $g_v$;
    \item for each $j \in [k+1]$, $w_j$ and $r_j$ value the dummy item $d_j$, and every label item except their own.
\end{itemize}

\cref{tab:binary-valuation-matrix} provides a simple visualization of the constructed instance $\mathcal{I(S)}$ of $\efconn$ and $\efopt$.

\begin{table}[h]
\centering
\footnotesize
\setlength{\tabcolsep}{3pt}
\renewcommand{\arraystretch}{1.15}
\resizebox{\linewidth}{!}{%
\begin{tabular}{|c|c|c|c|c|c|c|}
\hline
\multicolumn{1}{|>{\columncolor{gray!20}\centering\arraybackslash}m{2.5cm}|}
{\diagbox[width=7em]{\textsc{Agents}}{\textsc{Items}}}
& \multicolumn{1}{>{\columncolor{blue!12}\centering\arraybackslash}m{3.0cm}|}
{\makecell{$\{g_v \mid v \in V\}$}}
& \multicolumn{1}{>{\columncolor{blue!12}\centering\arraybackslash}m{2.6cm}|}
{\makecell{$\{d_j \mid j \in [k+1]\}$}}
& \multicolumn{1}{>{\columncolor{blue!12}\centering\arraybackslash}m{2.8cm}|}
{\makecell{$\{\ell_v \mid v \in V\}$}}
& \multicolumn{1}{>{\columncolor{blue!12}\centering\arraybackslash}m{1.3cm}|}
{$\ell_c$}
& \multicolumn{1}{>{\columncolor{blue!12}\centering\arraybackslash}m{2.8cm}|}
{\makecell{$\{\ell_e \mid e \in E\}$}}
& \multicolumn{1}{>{\columncolor{blue!12}\centering\arraybackslash}m{3.4cm}|}
{\makecell{$\{\ell_j,\ell'_j \mid j \in [k+1]\}$}}
\\
\hline
\multicolumn{1}{|>{\columncolor{orange!18}\centering\arraybackslash}m{2.5cm}|}
{\makecell{\textsc{Owner }\\\textsc{in $X$}}}
& \makecell{the corresponding\\$p_v$}
& $s$
& \makecell{the corresponding\\$p_v$}
& $c$
& \makecell{the corresponding\\$a_e$}
& \makecell{the corresponding\\$w_j,r_j$}
\\
\hline
\multicolumn{1}{|>{\columncolor{green!12}\centering\arraybackslash}m{2.5cm}|}
{\makecell{$p_v$\\$(v \in V)$}}
& $\mathbf{0}$
& $\mathbf{0}$
& $\mathbf{0}$
& $0$
& $\mathbf{0}$
& $\mathbf{0}$
\\
\hline
\multicolumn{1}{|>{\columncolor{green!12}\centering\arraybackslash}m{2.5cm}|}{$c$}
& $\mathbf{1}$
& $\mathbf{1}$
& $\mathbf{0}$
& $0$
& $\mathbf{0}$
& $\mathbf{0}$
\\
\hline
\multicolumn{1}{|>{\columncolor{green!12}\centering\arraybackslash}m{2.5cm}|}{$s$}
& $\mathbf{0}$
& $\mathbf{1}$
& $\mathbf{0}$
& $0$
& $\mathbf{0}$
& $\mathbf{0}$
\\
\hline
\multicolumn{1}{|>{\columncolor{green!12}\centering\arraybackslash}m{2.5cm}|}
{\makecell{$a_e$\\$(e=(u,v)\in E)$}}
& $\mathbf{e}_u+\mathbf{e}_v$
& $\mathbf{0}$
& $\mathbf{0}$
& $0$
& $\mathbf{0}$
& $\mathbf{0}$
\\
\hline
\multicolumn{1}{|>{\columncolor{green!12}\centering\arraybackslash}m{2.5cm}|}
{\makecell{$w_j$\\$(j \in [k+1])$}}
& $\mathbf{0}$
& $\mathbf{e}_j$
& $\mathbf{1}$
& $1$
& $\mathbf{1}$
& \makecell{$\mathbf{1}$, except at $\ell_j$}
\\
\hline
\multicolumn{1}{|>{\columncolor{green!12}\centering\arraybackslash}m{2.5cm}|}
{\makecell{$r_j$\\$(j \in [k+1])$}}
& $\mathbf{0}$
& $\mathbf{e}_j$
& $\mathbf{1}$
& $1$
& $\mathbf{1}$
& \makecell{$\mathbf{1}$, except at $\ell'_j$}
\\
\hline
\end{tabular}%
}
\caption{Block representation of the additive binary valuations in the reduction from \textsc{Independent Set}, together with the initial ownership of each item family. Rows are indexed by agents and columns are indexed by sets of items. The first row denotes the agent whose bundle contains the item(s) in $X$ corresponding to the respective column. The bold numbers are vectors of dimension same as the number of items corresponding to the respective column. In particular, $\mathbf{0}$ is the vector of all $0$ values, $\mathbf{e}_u$ denotes a vector with entry corresponding to $u$ being $1$ and the rest of the entries $0$. }
\label{tab:binary-valuation-matrix}
\end{table}

The construction of the instance $\mathcal{I(S)}$ described in \cref{tab:binary-valuation-matrix} can clearly be carried out in polynomial time. 
We denote the initial allocation by $X$ and it is specified in the second row of \cref{tab:binary-valuation-matrix}. Note that $X$ is near-$\mathsf{EF1}$ (\cref{def:near-ef1}), as only $c$ $\mathsf{EF1}$-envies $s$. It is easy to check that this is the only $\mathsf{EF1}$-envy in $X$.

We next show that no valid transfer can move a dummy item or a label item.

\begin{claim}\label{clm:adbin-dummy}
Any sequence of valid transfers starting from $X$ cannot involve the transfer of any dummy items $d_1,\ldots,d_{k+1}$.
\end{claim}
\begin{proof}
Fix any $j \in [k+1]$. Suppose that $d_j$ is the first among the dummy/label items that is transferred in some valid restoration sequence. 
If $d_j$ is transferred to an agent $x \not\in \{w_j,r_j\}$, then $w_j$ values both $d_j$ and the label item in the bundle of $x$, while it values its own bundle at $0$; hence, $w_j$ {\em develops} $\mathsf{EF1}$-envy towards $x$, and hence, this cannot be a valid transfer. If $d_j$ is moved to $w_j$, then $r_j$ develops an $\mathsf{EF1}$-envy towards $w_j$; similarly, if $d_j$ is moved to $r_j$, then $w_j$ develops an $\mathsf{EF1}$-envy towards $r_j$. Hence, no dummy item can be transferred from $X$, provided the label items are frozen.
\end{proof}

Now, we show a similar claim for the label items.
\begin{claim}\label{clm:adbin-label}
Any sequence of valid transfers starting from $X$ cannot involve a transfer of any label item.
\end{claim}
\begin{proof}
Let there be a label item $\ell$ such that the transfer of $\ell$ is a valid transfer with respect to $X$, and is the first among the dummy/label items that is transferred in some valid restoration sequence.
If $\ell$ is transferred to $s$, then there exists some watcher agent $w_j$ or $r_j$ for some $j \in [k+1]$ who values $\ell$,
and hence, values agent $n$'s new bundle at $2$. Since this watcher agent values her own bundle only at $0$, such a transfer is not valid with respect to $X$.

If $\ell$ is transferred to an agent other than $s$, then the destination bundle contains two label items. Since $k \geq 1$, there are at least four watcher agents in total, and at most two of them fail to value one of these two labels. Hence, some watcher agent values both labels while valuing its own bundle at $0$. This will create new $\mathsf{EF1}$-envy, and hence, such a transfer is not valid with respect to $X$. Overall, no label item constitutes a valid transfer, which proves the stated claim.
\end{proof}

From Claims~\ref{clm:adbin-dummy} and \ref{clm:adbin-label}, we conclude that no dummy item or label item can be transferred in any valid transfer with respect to $X$. We are now equipped to analyze our reduction.\\

\noindent{\bf \textsc{Independent Set} $\implies$ $\efopt$:}

First, suppose $G$ has an independent set $S \subseteq V$ with $\abs{S} \geq k$. Recall that, agent $c$ $\mathsf{EF1}$-envies agent $s$ while all other agents are $\mathsf{EF1}$-happy in the initial allocation $X$. So, we transfer the vertex items that correspond to any $k$ vertices of $S$, one by one, from their respective agents $p_v$ to the agent $c$. Each such transfer is valid: the envy of $c$ towards $s$ weakly decreases, watcher agents do not value the vertex items and no edge agent becomes $\mathsf{EF1}$-unhappy since $S$ is an independent set. In the resulting allocation, agent $c$ holds $k$ vertex items (that it values), while $s$ now holds all $k+1$ dummy items. Hence, $c$ is $\mathsf{EF1}$-happy towards $s$. Moreover, no edge agent $\mathsf{EF1}$-envies $c$, since $c$ contains no pair of adjacent vertices. All other agents are trivially $\mathsf{EF1}$-happy. Therefore, after the above $k$ valid transfers, the resulting allocation is $\mathsf{EF1}$ and so, $\mathcal{I(S)}$ is a \textsc{Yes}-instance of $\efopt$.\\

\noindent{\bf $\efopt$ $\implies$ \textsc{Independent Set}:}

Conversely, suppose that $\mathcal{I(S)}$ is a \textsc{Yes}-instance of $\efopt$. Then, there exists a valid transfer sequence that starts from $X$ and leads to an $\mathsf{EF1}$ allocation $Y$, using at most $k$ valid transfers. It follows from Claims~\ref{clm:adbin-dummy} and \ref{clm:adbin-label} that the only valid transfers starting from $X$ are the transfers of vertex items from $p_v$ to $c$, such that $c$ does not receive two vertex items that correspond to the endpoints of an edge in $G$. Also, transfer of vertex items does not make the transfer of any dummy item or a label item valid, as the same argument in Claims~\ref{clm:adbin-dummy} and \ref{clm:adbin-label} applies. Therefore, $Y_s$ must contain all the dummy items. Hence, $Y_c$ must have a value at least $k$, and hence must contain at least $k$ vertex items.

Finally, $Y_c$ must correspond to an independent set in $G$. This is because, for any $e = (u,v) \in E$, if both $g_u$ and $g_v$ belong to $Y_c$, then the agent $a_e$ would $\mathsf{EF1}$-envy $c$, a contradiction. Therefore, the set of vertices that correspond to vertex items held by $c$ in $Y$ is independent. Since it has size at least $k$, the graph $G$ contains an independent set of size at least $k$. That is, $(G,k)$ is a \textsc{Yes}-instance of the \textsc{Independent Set} instance $\mathcal{S}$. One can also check that $G$ has an independent set of size at least $k$ if and only if there is a sequence of valid transfers starting from $X$ that restores $\mathsf{EF1}$ by the same arguments. Therefore, the problem of deciding if $\efconn$ is possible, and the $\efopt$ problem, are $\mathrm{NP}$-hard, even when all the agents have additive binary valuations.
\end{proof}

\begin{remark}
As an immediate consequence of part (ii) of \hyperref[thm:EF1-Restoration-NP-hard-additive-binary-restated]{\Cref*{thm:EF1-Restoration-NP-hard-additive-binary}}, the optimization problem of computing the minimum number of valid transfers required to reach an $\mathsf{EF1}$ allocation is $\mathrm{NP}$-hard even when all the agents have additive binary valuations, and the initial allocation is near-$\mathsf{EF1}$.
\end{remark}


\subsection{\texorpdfstring{$\efconn$}{EF1-Restoration} for Orientations under Graphical Valuations} \label{subsec:graphical}

In this section, we consider the valuation class defined using multigraphs with additive valuations where an item is valued by at most two agents and hence can be modeled as an edge between them. We study a special class of non-wasteful allocations, called \emph{orientations}. Note that $\mathsf{EF1}$ orientations always exist for monotone valuations \cite{deligkas2024ef1efxorientations}. In our $\efconn$ problem, starting from an arbitrary orientation, the goal is to reach an $\mathsf{EF1}$ orientation by only making valid transfers, and traversing only via orientations.

We consider \emph{binary} valuations, where an item $g$ that corresponds to an edge $g = (i,j) \in E(i,j)$ between two agents $i$ and $j$ will be valued as $v_i(g) = v_j(g) \ceq 1$, while $v_k(g) \ceq 0$ for all agents $k \neq i,j$. We prove that $\efconn$ is possible in this setting (\cref{thm:binary} and its generalization, \cref{thm:binary-arb}). We do so by developing \Cref{algo:graphical-binary}, that achieves an optimal number of transfers (in the worst case) to reach an $\mathsf{EF1}$ allocation.

We begin make some useful observations regarding orientations and the envy graph (see \cref{def:envy-graph}) corresponding to them, for these valuation classes.

\begin{lemma}\label{lem:basic_orientation}
Consider a fair division instance with additive binary valuations, and let $X$ be an orientation. Then, we have
\begin{enumerate}
    \item \label{itm:own-bundle} For any $i,j \in [n]$, $v_i(X_i) = |X_i| \geq v_j(X_i)$. That is, if the allocation is an orientation, every agent values their bundle at least as much as any other agent values it.
    \item \label{itm:d-length} If there exists a (directed) path  $(i_0,i_1, i_2, \ldots, i_d)$ of length $d$ from agent $i_0$ to agent $i_d$ in the envy graph $\genvy(X)$, then $v_{i_d}(X_{i_d}) \geq v_{i_0}(X_{i_0}) + d$. 
    \item \label{itm:envy-graph-acyclic} The envy graph $\genvy(X)$ is acyclic.
\end{enumerate}
\end{lemma} 
\begin{proof}

(1) Since $X$ is an orientation, and the valuation is binary, we have $v_i(X_i)=|X_i|$. Note that, for any other agent $j \neq i$, we have $v_j(g)\in \{0,1\}$ for any $g \in X_i$. Hence, we have $v_j(X_i)\leq |X_i|=v_i(X_i)$.

(2) Assume agent $i$ envies agent $j$. Then, using \Cref{itm:own-bundle}, we have $|X_i|=v_i(X_i)<v_i(X_j)\leq v_j(X_j)=|X_j|$. Therefore, we obtain $v_j(X_j)\geq v_i(X_i)+1$. Generalizing it to a $d$-length envy path in $\genvy(X)$, say  $(i_0, i_1, \ldots, i_d)$, we obtain the following: 
\[
v_{i_d}(X_{i_d}) \geq v_{i_{d-1}}(X_{i_{d-1}})+1
\geq v_{i_{d-2}}(X_{i_{d-2}})+2 \geq
\ldots
\geq v_{i_0}(X_{i_0})+d.
\]
(3) Let us now assume a cycle $(i_0,\ldots,i_d)$ in $\genvy(X)$, where agent $i_d$ envies $i_0$. Using \Cref{itm:d-length}, we have $v_{i_d}(X_{i_d}) \geq v_{i_0}(X_{i_0})+d$. And then using part (1), we have $v_{i_0}(X_{i_0})\geq v_{i_d}(X_{i_0})$. Combining the last two inequalities, we obtain  $v_{i_d}(X_{i_d})\geq v_{i_d}(X_{i_0})$. This means $i_d$ does not envy $i_0$, leading to a contradiction. Hence, $\genvy(X)$ is acyclic.
\end{proof}

\cref{lem:basic_orientation} can be easily extended to the case of pairwise-homogeneous graphical valuations.

\begin{corollary} \label{cor:homo-acyclic}
    For any fair division instance with pairwise-homogeneous graphical valuations, the envy graph $\genvy(X)$ is acyclic for any orientation $X$.
\end{corollary}


\subsubsection{Binary Graphical Valuations with Near-\texorpdfstring{$\mathsf{EF1}$}{EF1} Orientations} \label{subsubsec:graphical-binary-near-EF1}

We begin by describing an $\efconn$ algorithm (\Cref{algo:graphical-binary-basic}) when the input orientation is \emph{near-$\mathsf{EF1}$}, i.e., only one agent (agent $1$, without loss of generality) is not $\mathsf{EF1}$-happy. The case of an arbitrary orientation is analogous, and is discussed in \cref{subsubsec:graphical-binary-arb}.

\begin{algorithm}[htbp]
\caption{\texorpdfstring{$\efconn$}{EF1-Restoration} of a near-\texorpdfstring{$\mathsf{EF1}$}{EF1} orientation on a multigraph} \label{algo:graphical-binary-basic}
\SetAlgoLined
\SetKwInOut{Input}{Input}
\SetKwInOut{Output}{Output}
\SetKwFunction{FMain}{EF1\_Orientation\_Restorer\_Basic}
\SetKwProg{Fn}{}{:}{}
\SetKw{KwTo}{to}
\SetKw{KwRet}{return}
\DontPrintSemicolon

\Input{A fair division instance $\mathcal{I}=\left<[n], \G, \{v_i\}_{i \in [n]} \right>$ on a multigraph with additive binary valuations, and a near-$\mathsf{EF1}$ orientation $X = (X_1, \ldots, X_n)$.}
\Output{ An $\mathsf{EF1}$ orientation.}
\BlankLine

\Fn{\FMain{$X$}}{
    $\mathcal{E} \leftarrow \text{agents $\mathsf{EF1}$-envied by agent } 1 \text{ in } X$\;
    \While{$S \neq \emptyset$}{
        $s \leftarrow \text{closest sink to $1$ in $\genvy(X)$}$\;
        $\mathcal{P} \leftarrow \text{a shortest path from $1$ to $s$ in $\genvy(X)$}$\;
        $p \leftarrow \text{the predecessor of $s$ on $\mathcal{P}$}$\;
        $g \leftarrow \text{an item in $X_s$ with $v_{p}(g) = 1$}$\;
\BlankLine
\makebox[0.8\linewidth][l]{$\left.\begin{array}{l}
    X_s \leftarrow X_s \sm g \\
    X_p \leftarrow X_p \cup g
\end{array}
\right\} \: \text{Transfer of } g \text{ from } s \text{ to } p$}
\BlankLine
    Update $\mathcal{E}$\;
}
\KwRet{$X = (X_1, \ldots, X_n)$}\;
}
\end{algorithm}

\begin{theorem} \label{thm:binary}
For the $\efconn$ problem on a multigraph with additive binary valuations, \Cref{algo:graphical-binary-basic} terminates with an $\mathsf{EF1}$-orientation, and it makes at most $q(n - 1)$ valid transfers of items, where $q$ is the number of items that agent $1$ requires to become $\mathsf{EF1}$-happy in the input near-$\mathsf{EF1}$ orientation $X$.
\end{theorem}

\begin{proof}
Let $X$ be the input near-$\mathsf{EF1}$ orientation. We require $q$ valid transfers (see \cref{def:valid-transfer}) to be made to agent $1$ to restore $\mathsf{EF1}$. We claim that \Cref{algo:graphical-binary-basic} makes at most $(n-1)$ valid transfers in the system, per item that agent $1$ receives.

From \cref{lem:basic_orientation}, we know that $\genvy(X)$ is acyclic. Hence, a sink in $\genvy(X)$ can transfer an item without creating any new $\mathsf{EF1}$-envy. However, agent $1$ may not positively value any item in the bundle of any sink. Therefore, in \Cref{algo:graphical-binary-basic}, we transfer an item from a sink to one of its predecessors in $\genvy(X)$. Since any predecessor must positively value at least one item in the sink's bundle, the transfer of such an item maintains an orientation. We claim that the allocation that results from such a transfer is a near-$\mathsf{EF1}$ orientation. Furthermore, we show that agent $1$ will receive an item (that she values) after at most $n-1$ such (valid) transfers. Hence, after $q(n-1)$ transfers, we must reach an $\mathsf{EF1}$ orientation.

Among all sinks reachable from $1$ in $\genvy(X)$, let $s$ be a sink that is closest to $1$. We note that $s \neq 1$ as, otherwise, $X$ is already an $\mathsf{EF1}$-orientation. Let $\mathcal{P}$ be a shortest path from $1$ to $s$, say of length $d$. Let $p$ be the predecessor of $s$ on $\mathcal{P}$. Since $p$ envies $s$, there exists an item $g \in X_s$ such that $v_p(g) = 1$. Moreover, since $X$ is an orientation, we know $v_s(g) = 1$ as well. Hence, after this transfer, the resulting allocation ($Y$, say) remains an orientation.

We have $Y_s = X_s \sm g, Y_p = X_p \cup g$, and $Y_j = X_j$ for all $j \neq p,s$. Since $g$ is valued only by $p$ and $s$, no agent other than $s$ can have any $\mathsf{EF1}$-envy towards $p$ after this transfer. Using \cref{lem:basic_orientation}, we have $v_s(X_s) \geq v_p(X_p)+1$. Hence,
\[
   v_s(Y_s) = v_s(X_s) - 1 \geq v_p(X_p) = v_p(Y_p \sm g)
\]
That is, $s$ has no $\mathsf{EF1}$-envy towards $p$ in $Y$. Therefore, the transfer of $g$ from $s$ to $p$ is valid. 

Note that if $p \neq 1$, $p$ has no $\mathsf{EF1}$-envy towards anyone, including $s$, in $X$. For any agent $t$ who $p$ envies in $X$, we have $v_p(X_p)\geq v_p(X_t)-1$, and $v_p(Y_p)=v_p(X_p)+1\geq v_p(X_t)=v_p(Y_t)$. That is, $p$ does not envy anyone in $Y$. Therefore, $p$ becomes a sink in $\genvy(Y)$. 
 
Observe that, in $Y$, $p$ is a sink which has a path of length $d - 1$ from $1$ in $\genvy(Y)$. Thus, after each valid transfer from a closest sink to its predecessor, the distance between agent $1$ to its closest sink in the envy graph decreases by $1$. Since the closest sink in $\genvy(X)$ can be at distance at most $n-1$, it follows that at most $n - 1$ transfers from a closest sink to its predecessor suffice for agent $1$ to receive one item. Since the recipient always values the transferred item positively, the orientation property is maintained. Overall, $q$ items are transferred by \Cref{algo:graphical-binary-basic} to agent $1$, by making at most $q(n-1)$ valid transfers, and this results in an $\mathsf{EF1}$ orientation. This completes the proof.
\end{proof}


\subsubsection{Binary Graphical Valuations with Arbitrary Orientations} \label{subsubsec:graphical-binary-arb}

In this section, we modify \Cref{algo:graphical-binary-basic}, and design an $\efconn$ algorithm (\Cref{algo:graphical-binary}) for an arbitrary input orientation, and prove the following theorem that generalizes \cref{thm:binary}.

\thmgraphicalbinary*
\label{thm:graphical-binary-restated}

\begin{algorithm}[htbp]
\caption{\texorpdfstring{$\efconn$}{EF1-Restoration} on multigraphs} \label{algo:graphical-binary}
\SetAlgoLined
\SetKwInOut{Input}{Input}
\SetKwInOut{Output}{Output}
\SetKwFunction{FMain}{EF1\_Orientation\_Restorer}
\SetKwProg{Fn}{}{:}{}
\SetKw{KwTo}{to}
\SetKw{KwRet}{return}
\DontPrintSemicolon

\Input{A fair division instance $\mathcal{I}=\left<[n], \G, \{v_i\}_{i \in [n]} \right>$ on a multigraph with additive binary valuations, and an orientation $X = (X_1, \ldots, X_n)$.}
\Output{ An $\mathsf{EF1}$ orientation.}
\BlankLine

\Fn{\FMain{$X$}}{
    $\mathcal{U} \leftarrow \text{agents who are $\mathsf{EF1}$-unhappy in } X$\;
    \While{$\mathcal{E} \neq \emptyset$}{
        $s \leftarrow \text{a sink in $\genvy(X)$ who is closest to some agent in $\mathcal{U}$}$\;
        $i \leftarrow \text{an agent in $\mathcal{U}$ who is closest to $s$, in $\genvy(X)$}$ \label{step:choice-of-i-graphical-general}
        $\mathcal{P} \leftarrow \text{a shortest path from $i$ to $s$ in $\genvy(X)$}$\;
        $p \leftarrow \text{the predecessor of $s$ on $\mathcal{P}$}$\;
        $g \leftarrow \text{an item in $X_s$ with $v_{p}(g) = 1$}$\;
\BlankLine
\makebox[0.8\linewidth][l]{$\left.\begin{array}{l}
    X_s \leftarrow X_s \sm g \\
    X_p \leftarrow X_p \cup g
\end{array}
\right\} \: \text{Transfer of } g \text{ from } s \text{ to } p$}
\BlankLine
    Update $X$ and $\mathcal{U}$\;
}
\KwRet{$X = (X_1, \ldots, X_n)$}\;
}
\end{algorithm}

\begin{proof}
Let $X$ be the input orientation, and let $\mathcal{U}$ be the set of $k$ agents who are $\mathsf{EF1}$-unhappy in $X$.

From \cref{lem:basic_orientation}, we know that $\genvy(X)$ is acyclic and a sink in $\genvy(X)$ can transfer an item without creating any new $\mathsf{EF1}$-envy. However, none of the agents in $\mathcal{U}$ may positively value any item in the bundle of any sink. Therefore, in \Cref{algo:graphical-binary}, we transfer an item $g$ from a sink $s$ to one of its predecessors $p$ in $\genvy(X)$, with $v_p(g) = v_s(g) = 1$. By imitating the proof of \cref{thm:binary}, it is easy to see that such a transfer is valid (i.e., creates no new $\mathsf{EF1}$-envy), and maintains an orientation.

Note that, if  $i \in \mathcal{U}$ is the closest agent to some sink $s$ in $\genvy(X)$, then (by our choice of $i$ in \linkline{step:choice-of-i-graphical-general}), every agent $j \neq i$ in a shortest path $\mathcal{P}$ from $i$ to $s$ must be $\mathsf{EF1}$-happy. Furthermore, $p$ must be a sink in the envy graph $\genvy(Y)$, where $Y$ is the allocation that results from the transfer of $g$ from $s$ to $p$ in $X$. This again follows by imitating the proof of \cref{thm:binary}.

Observe that, during the first iteration of the algorithm, the distance between $i$ and $s$ (as picked by \Cref{algo:graphical-binary}) is at most $(n-k)$. In general, if there are $\ell \leq k$ $\mathsf{EF1}$-unhappy agents at any point in the execution of \Cref{algo:graphical-binary}, the shortest distance between some agent in $\mathcal{U}$ and some sink in the envy graph is at most $(n - \ell)$. Hence, \Cref{algo:graphical-binary} will make at most $(n - k) + (n - (k - 1)) + \ldots + (n - 1) \leq nk-\frac{k^2}{2}$ transfers to restore $\mathsf{EF1}$, assuming that each of the agents in $\mathcal{U}$ requires only one item to become $\mathsf{EF1}$-happy. If $q$ is the maximum number of items that any agent in $\mathcal{U}$ requires to become $\mathsf{EF1}$-happy, then \Cref{algo:graphical-binary} makes at most $qk(n - \frac{k}{2})$ transfers before it reaches an $\mathsf{EF1}$ orientation.
\end{proof}

Next, we provide an example to show that the bounds on transfers given in \cref{thm:binary} and \cref{thm:binary-arb} are tight up to an additive factor that depends only on $k,q$ and is independent of n. In particular, we prove the following result.

\begin{theorem} \label{thm:graphical-binary-lower-bound}
For a fair division instance on $n$ agents with graphical valuations that are additive and binary, $\efconn$ requires at least $qk(n-k)$ valid transfers to be made in the worst case, where $k$ is the number of $\mathsf{EF1}$-unhappy agents in the input orientation, and $q$ is the maximum number of items that any $\mathsf{EF1}$-unhappy agent in the input orientation requires to become $\mathsf{EF1}$-happy. 
\end{theorem}

\begin{proof}
Consider a fair division instance in which the underlying graph $\gval$ (corresponding to the graphical valuation) is defined as follows -- agents in the set $[k]$ all form a \enquote{star-like} structure and are all connected to agent $k+1$, while the agents $k+1 - \cdots - n$ form a path. Consider the initial allocation $X$ as follows: $X_i \ceq \emptyset$ for $i \in [k]$, $|X_{k+1}| \ceq 2k$, and $|X_j| \ceq |X_{j-1}|+1$ for each $j \in \{k+2, k+3, \ldots, n\}$. Let us now define the valuations as follows. Let $X_{k+1}=\{g_1, g_2, \ldots, g_{2k}\}$. For $i \in [k]$, we define her valuation $v_i$ as
\[ v_i(g_{\ell}) \ceq
\begin{cases} 
      1, & \text{if } \ell \in \{2i-1,2i\}  \\
      0, & \text{otherwise}
   \end{cases}
\]

For any agent $i \in \{k+1, k+2, \ldots, n-1\}$, we define her valuation $v_i$ as 
\[ v_i(g) \ceq
\begin{cases} 
      1, & \text{if } g \in X_i \cup X_{i+1}   \\
      0, & \text{otherwise}
   \end{cases}
\]
Finally, for agent $n$, we define $v_n(g) \ceq 1$ if and only if $g\in X_n$. 
The valuation is graphical since, for each $\ell \in [2k]$, $g_{\ell}$ is valued only by the agents $\ceil{\frac{\ell}{2}}$ and $k+1$, whereas an item in the bundle $X_i$, $i \geq k+2$ is valued only by the agents $i-1$ and $i$. Clearly, $X$ is an orientation, since every agent values all the items in her bundle.

It is easy to see that there are exactly $k$ $\mathsf{EF1}$-unhappy agents in $X$, namely the agents in the set $[k]$, all of whom $\mathsf{EF1}$-envy agent $k+1$. Note that any agent $i \in [k]$ values exactly two items in the bundle $X_{k+1}$ (and no other item). Therefore, each agent $i \in [k]$ must receive one item from $X_{k+1}$, and hence, agent $k+1$ must lose one item that she values. Note that any agent $j \in \{k+2, k+3, \ldots, n\}$ values their left neighbor's bundle $X_{j-1}$ at zero. Since we need to maintain an orientation, any item transfer must therefore be from right to left. Now, if agent $k+1$ simply transfers an item to agent $i \in [k]$, it creates $\mathsf{EF1}$ envy from $k+1$ to $k+2$. Therefore, agent $k+1$ must get an item from agent $k+2$. By our construction, no agent $j \in \{k+2, \ldots, n-1\}$ can transfer any item to $j-1$ without $\mathsf{EF1}$-envying towards $j+1$. Hence, agent $n$ must transfer an item to $n-1$, so that $n-1$ can give an item to $n-2$, and so on. That is, for any agent $i \in [k]$ to receive a single item, it requires at least $(n-k)$ valid transfers. Therefore, we require at least $k(n-k)$ valid transfers overall, to reach an $\mathsf{EF1}$ orientation.

It is easy to modify the instance described above for the setting where an $\mathsf{EF1}$-unhappy agent requires $q$ items to be $\mathsf{EF1}$-happy. We can set $|X_{k+1}| \ceq 2qk$, where agent $i \in [k]$ now values $2q$ items in $X_{k+1}$, distinct from those valued by any $j \in [k] \sm i$. This will ensure that we require at least $qk(n-k)$ valid transfers to reach an $\mathsf{EF1}$ orientation. This completes the proof.
\end{proof}


\section{\texorpdfstring{$\mathrm{PSPACE}$}{PSPACE}-completeness of \texorpdfstring{$\efconn$}{EF1-Restoration}}
\label{sec:pspace}

In this section, we show that the $\efconn$ problem may not always have an affirmative answer, and in fact, it is $\ps$-complete to decide so, even with \emph{monotone binary} valuation functions, the input allocation being near-$\mathsf{EF1}$, and even when \emph{valid exchanges} of single items are permitted, in addition to valid transfers. Our reduction is inspired by Igarashi et al. \cite{igarashi2024reachability}, in which they give a polynomial-time reduction from the $\ps$-complete problem of $\pmr$ \cite{BonamyBHIKMMW19} to the problem of deciding if there exists an $\mathsf{EF1}$ exchange path between two fixed $\mathsf{EF1}$ allocations, with general additive valuations.

Their problem involves deciding if there is a sequence of exchanges between two given $\mathsf{EF1}$ allocations. We would like to prove that deciding if \emph{any} $\mathsf{EF1}$ allocation (as opposed to some fixed $\mathsf{EF1}$ allocation, as in \cite{igarashi2024reachability}) is reachable from a given near-$\mathsf{EF1}$ allocation via transfers \emph{and} exchanges (as opposed to reachability only via exchanges in \cite{igarashi2024reachability}) is $\ps$-complete, provided we maintain a near-$\mathsf{EF1}$ allocation after each operation. This leads to an involved construction of the valuation functions in our $\efconn$ instance, so as to ensure that \emph{every} $\mathsf{EF1}$ allocation that may be constructed from our instance corresponds to \emph{exactly one} perfect matching - the target perfect matching in the $\pmr$ instance with which we begin.

\pspace*

\begin{proof}
We first show that $\efconn \in \ps$ for monotone binary valuations. Then, we prove $\ps$-hardness by reducing from the $\pmr$ problem.

\paragraph{Containment in $\ps$:} We first show that the $\efconn$ problem is in $\ps$. Recall that $\ps$ is the set of all decision problems that can be solved by a deterministic polynomial-space Turing machine. We can solve the $\efconn$ problem nondeterministically by simply guessing a path from the given near-$\mathsf{EF1}$ allocation to some $\mathsf{EF1}$ allocation. Since the total number of allocations is at most $n^m$ , if there exists such a path, then there exists one with length at most $n^m$. This shows that the problem is in $\nps$, the nondeterministic analogue of $\ps$.\footnote{$\nps$ is the set of all decision problems that can be solved by a nondeterministic polynomial-space Turing machine.} It is known that $\nps = \ps$ \cite{Savitch70}, which implies that the $\efconn$ problem is in $\ps$. We refer the reader to a standard complexity theory text like \cite{AroraBarak} for formal definitions of $\ps$, $\nps$, etc.

\paragraph{$\ps$-hardness:} We reduce the $\ps$-complete $\pmr$ problem to $\efconn$. 
    
\begin{definition}[The $\pmr$ problem] \label{def:perfect-matching-reconfiguration}
    For an undirected bipartite graph $G=(A\sqcup B,E)$ with $|A|=|B|=n$, the problem is to decide reachability between two given perfect matchings $M_0$ and $M^*$. It involves deciding if $M_0$ and $M^*$ can be reached from each other via a sequence of perfect matchings $M_0, M_1, M_2, \ldots, M_t=M^*$, such that for each $k \in [t]$, there exist edges $e^k_1, e^k_2, e^k_3, e^k_4$ of $G$ such that $M_{k-1} \sm M_k = \{e^k_1, e^k_3\}$, $M_k \sm M_{k-1} = \{e^k_2, e^k_4\}$, and $e^k_1, e^k_2, e^k_3, e^k_4$ form a cycle.
\end{definition}
    
The operation of going from $M_{k-1}$ to $M_k$ is called a \emph{flip}, and we say that $M_{k-1}$ and $M_k$ are \emph{adjacent} to each other.  

{\em Construction of an instance of the $\efconn$ problem:}
For a given instance of \textsf{Perfect} \textsf{Matching} \textsf{Reconfiguration} as described above, let us denote $A \ceq \{a_1, a_2, \ldots, a_n\}$ and $B \ceq \{b_1, b_2, \ldots, b_n\}$. We write $N(v)$ to denote the set of neighbors of vertex $v \in A \cup B$ in $G$. By possible renaming of vertices, let $M_0 \ceq \{(a_i,b_i)\}_{i \in [n]}$ and the final matching $M^* \ceq \{(a_i, b_{\pi(i)})\}_{i \in [n]}$, where $\pi$ is a permutation on $[n]$.

Now, we construct an instance of the $\efconn$ problem, with agents having monotone binary valuations.

    \begin{itemize}
        
    \item Create a set $\mathcal{N} \ceq \set{0, \ldots, n+2}$ of $n+3$ agents and a set $\G=\{a_i, \bar{a}_i,b_i\mid i\in [n]\}\cup\{r_1,r_2,r_3,r_4\}$ of $3n+4$ items. That is, for each $i \in [n]$, we have three items $a_i, \bar{a}_i$, and $b_i$, with  four additional items $r_1$, $r_2$, $r_3$, and $r_4$.
    
    \item Create the initial near-$\mathsf{EF1}$ allocation $X = \set{X_0, \ldots, X_{n+2}}$, that is given as an input, from the matching $M_0$ of $G$ as $X_0 \coloneq \emptyset$, $X_i \coloneq \set{a_i,\bar{a}_i,b_i}$ for each $i \in [n]$, $X_{n+1} \coloneq \set{r_1, r_2}$, and $X_{n+2} \coloneq \set{r_3,r_4}$.
   
    \item We define the valuations\footnote{When we say $a_i/\bar{a}_i$ while defining the valuations, we mean to say one of $a_i$ or $\bar{a}_i$ in included.} as follows:
\[
v_0(S) \coloneq 
\begin{cases}
      1 & \text{if } S = \set{a_i, \bar{a}_i, b_j} \ \forall i,j \in [n], \\
      1 & \text{if } S \subsetneq \set{a_i, \bar{a}_i, b_j}, \text{ with } \abs{S} = 2, \\
        & \quad \forall i \in [n], \ \forall j \in [n] \setminus \pi(i), \\
      0 & \text{if } S = \set{a_i/ \bar{a}_i, b_{\pi(i)}}, \ \forall i \in [n].
\end{cases}
\]

For each agent $i \in [n]$, we have $X_i = \set{a_i,\bar{a}_i,b_i}$ in $X$. We define the valuation for agent $i\in [n]$ as 
    
\[
v_i(S) \coloneq 
\begin{cases}
      1 & \text{if } S = \set{a_i, \bar{a}_i, b_j} \  \forall b_j \in N(a_i), \\
      0 & \text{if } S \subsetneq \set{a_i, \bar{a}_i, b_j} \ \forall b_j \in N(a_i), \\
      1 & \text{if } S = \set{r_3} \text{ or } S = \set{r_4}, \\
      0 & \text{if } S = \set{a_i/\bar{a}_i, b_j, c}, \forall c \in \set{b_k, a_k, \bar{a}_k, r_1, r_2},\\
        & \quad \forall k \in [n] \setminus \{i\}, \forall b_j \in N(a_i).
\end{cases}
\]

The valuation of agent $n+1$, where $X_{n+1} = \set{r_1,r_2}$ in $X$, is defined below.
    
\[
v_{n+1}(S) \coloneq 
\begin{cases} 
      1 & \text{if } S = \set{r_1, r_2}, \\
      0 & \text{if } S = \set{r_1} \text{ or } S = \set{r_2}, \\
      0 & \text{if } S = \set{r_1/r_2, c}, \forall c \in \set{r_3, r_4, b_i}, \ \forall i \in [n],\\
      1 & \text{if } S \subsetneq \set{a_i, \bar{a}_i, b_j}, \text{ with } \abs{S} = 2,\\
        & \quad \forall i \in [n], \ \forall j \in [n] \setminus \pi(i), \\
      0 & \text{if } S = \set{a_i/\bar{a}_i, b_{\pi(i)}}.
\end{cases}
\]

We define the valuation of agent $n+2$ as follows, who holds the bundle $X_{n+2} = \set{r_3,r_4}$ in $X$.

\[
v_{n+2}(S) \coloneq 
\begin{cases} 
      1 & \text{if } S = \set{r_3, r_4}, \\
      0 & \text{if } S = \set{r_3} \text{ or } S = \set{r_4}, \\
      1 & \text{if } S = \set{r_1} \text{ or } S = \set{r_2}, \\
      0 & \text{if } S = \set{r_3/r_4, c}, \forall c \in \set{a_i, \bar{a}_i, b_i}, \ \forall i \in [n].
\end{cases}
\]
\end{itemize}

Note that the valuations are listed above explicitly for a polynomial number of sets. We extend the valuations to other unlisted subsets monotonically, i.e., every superset of a $1$-valued bundle has a value $1$ and every subset of a $0$-valued bundle has a value $0$.  Clearly, this instance of $\efconn$ can be constructed in polynomial time.

{\em Properties of $X$: }Note that in the initial allocation $X$, agent $0$ has $\mathsf{EF1}$-envy towards all the agents $i \in [n]$ for whom $b_i\neq b_{\pi(i)}$. Agent $0$ is the unhappy agent in all of the near-$\mathsf{EF1}$ allocations in our instance. Observe that the initial allocation $X$ is near-$\mathsf{EF1}$, as no agent apart from agent $0$ has $\mathsf{EF1}$-envy towards any other agent. The goal is to perform a sequence of operations, i.e., valid transfers or exchanges, and finally reach some $\mathsf{EF1}$ allocation.  

Note that there is only one $\mathsf{EF1}$ allocation $Z$ in our instance, where $Z_i=\{a_i,\bar{a}_i,b_{\pi(i)}\}$ for each $i\in [n]$, $Z_0=\emptyset$, $Z_{n+1}=X_{n+1},$ and $Z_{n+2}=X_{n+2}$. These bundles  correspond to $M^*$ in the $\pmr$ problem.

The valuations forbid the following types of operations: 
   
\begin{itemize}
       
       \item \textit{No agent $i \in \mathcal{N}$ can transfer an item from her bundle}, because she would then $\mathsf{EF1}$-envy either agent $n+2$ (if $i \neq n+2$), or agent $n+1$ (if $i = n+2$).
       \item \textit{No agent $i \in [n]$ can exchange an item with agent $n+1$.} If agent $i$ gives away $a_i$ or $\bar{a}_i$ in exchange of $r_1$ or $r_2$, then she $\mathsf{EF1}$-envies agent $n+2$. Also, if $b_j \in X_i$ is exchanged with agent $n+1$ for $r_1$ or $r_2$, then the value of agent $n+1$'s bundle drops to $0$ and she $\mathsf{EF1}$-envies each agent $j \in [n]$ who does not possess $b_{\pi(j)}$.
       
       \item \textit{No agent $i \in [n]$ can exchange the $a_i$ or $\bar{a}_i$ in their bundle, with any agent $j$} because agent $i$ would then $\mathsf{EF1}$-envy agent $n+2$.
       \item \textit{No agent $i \in [n]$ can exchange any pair of items with agent $n+2$}, since agent $n+2$ would then $\mathsf{EF1}$-envy agent $n+1$.
       \item \textit{No agent $i \in [n]$ can exchange the $b_j$ in their current bundle with some other $b_{j'}$, if $b_{j'} \notin N(a_i)$}, for agent $i$ would then $\mathsf{EF1}$-envy agent $n+2$.
       \item \textit{Agent $n+1$ can not exchange any pair of items with agent $n+2$}, for agent $n+1$ would then $\mathsf{EF1}$-envy each agent $i \in [n]$ who does not possess $b_{\pi(i)}$.
       \item {\em Agent $0$ cannot receive any item.} This is because agent $0$ does not value the items $r_1,r_2,r_3,r_4$, and agent $i\in [n]$ cannot transfer an item as stated above.
   \end{itemize}

   Therefore, the only valid operation is an exchange of the {\em $b$-type} items between a pair of agents $i,j$. More precisely, let the near-$\mathsf{EF1}$ allocation at some point be $Y$, such that $Y_i=\{a_i,\bar{a}_i, b_k\}$ and $Y_j=\{a_j,\bar{a}_j, b_\ell\}$. The only valid operation is an exchange of $b_k$ and $b_\ell$ between agents $i,j$.
   
   Suppose, at some point, we have an allocation $Y$ where $Y_i = \set{a_i, \bar{a}_i, b_k}$ and $Y_j = \set{a_j, \bar{a}_j, b_\ell}$ for some $i, j \in [n]$. This represents (a part of) a perfect matching $M'$ in the graph $G$, where $(a_i, b_k), (a_j,b_\ell) \in M'$. An exchange of the items $b_k$ and $b_\ell$ between agents $i$ and $j$ corresponds to a \emph{flip} of the perfect matching $M'$ so that the edges $(a_i,b_k), (a_j,b_\ell)$ are replaced by $(a_i, b_\ell)$ and $(a_j,b_k)$ in the new perfect matching $M''$. Therefore, every valid operation in our instance corresponds to a flip in the perfect matching in the given instance. Moreover, for any intermediate near-$\mathsf{EF1}$ allocation $Y$, we may simply obtain the corresponding perfect matching as follows -- vertex $a_i$ is matched to vertex $b_j$ if $b_j \in Y_i$ for $i\in [n]$.
   
   Over the course of these valid operations, the unhappy agent $0$ will receive no item, and continue to hold an empty bundle. Therefore, the only way to make agent $0$ lose their $\mathsf{EF1}$-envy is to reach an allocation $X^*$ in which each agent $i \in [n]$ holds the bundle $X^*_i = \set{a_i, \bar{a}_i, b_{\pi(i)}}$. Observe that such an allocation $X^*$ corresponds to the perfect matching where there exists an edge between vertices $a_i$ and $b_{\pi(i)}$ for each $i \in [n]$. This is precisely the target perfect matching $M^*$. 
   
   Thus, there is a sequence of adjacent perfect matchings from $M_0$ to $M^*$ in $G$ if and only if there is a sequence of valid operations that transforms $X$ to the unique $\mathsf{EF1}$ allocation $Z$. This completes the reduction, and the proof.
   \end{proof}


\section{The \texorpdfstring{$\efxconn$}{EFX-Restoration} Problem} \label{sec:identical-efx}

It is natural to ask whether our results on $\efconn$ can be extended to stronger fairness notions like $\mathsf{EFX}$. The existence of $\mathsf{EFX}$ is guaranteed when agents have identical valuations on goods. However, we show that $\efxconn$ may not be possible, even for identical {\em additive} valuations with all the items being goods. We also show that it is $\mathrm{NP}$-hard to determine whether $\efxconn$ is possible on a given instance (\Cref{thm:EFX-Restoration-weakly-NP-hard-identical-additive}).

We now formally define the corresponding notions of $\mathsf{EFX}$-envy, $\mathsf{EFX}$ and near-$\mathsf{EFX}$ allocations, and state the $\efxconn$ problem.

\begin{definition}[$\mathsf{EFX}$-envy]\label{def:efx-envy}
We say an agent $i$ \emph{$\mathsf{EFX}$-envies}  agent $j$ in an allocation $X$ of goods if \emph{there exists} $g \in X_j$ such that $v_i(X_i) < v_i(X_j \sm g)$.

That is, $i$ continues to envy $j$ even after virtually \emph{removing a particular good} from $j$'s bundle. 
We say $i$ is \emph{$\mathsf{EFX}$-happy} in $X$ if she does not $\mathsf{EFX}$-envy any other agent.
\end{definition}

With this modified definition, we may define an $\mathsf{EFX}$ allocation and a near-$\mathsf{EFX}$ allocation.

\begin{definition}[Envy-freeness up to any item ($\mathsf{EFX}$) \cite{CaragiannisKMPS19}]\label{def:efx}
An allocation $X$ of a set of goods is said to be $\mathsf{EFX}$ if every agent is $\mathsf{EFX}$-happy in $X$.
\end{definition}

An agent $i$ is said to $\mathsf{EFX}$-envy an agent $j$ if there exists $g \in X_j$ such that $v(X_i) < v(X_j \sm g)$. An allocation is $\mathsf{EFX}$ if there no $\mathsf{EFX}$-envy in the system.

The notions of valid transfers, the $\efxconn$ and $\efxopt$ problems, and near-$\mathsf{EFX}$ allocations, are defined in a fashion analogous to \cref{def:valid-transfer}, \cref{prob:ef1-restoration}, \cref{prob:k-valid-transfers-EF1}, and \cref{def:near-ef1} respectively.

\begin{definition}[Valid Transfer (w.r.t.\ $\mathsf{EFX}$)]\label{def:valid-transfer-efx}
For an allocation $X$, consider a single item transfer from an agent to another, resulting in allocation $Y$. Then, this transfer is said to be \emph{valid} if no agent $\mathsf{EFX}$-envies another agent in $Y$ any more than they did in $X$.
\end{definition}

\begin{restatable}{problem}{efxrestoration}[{\upshape$\efxconn$}]
\label{prob:efx-restoration}
Given a fair division instance and an input allocation $X$, determine whether it is possible to reach an $\mathsf{EFX}$ allocation (from $X$) by a sequence of \emph{valid transfers}.
\end{restatable}

\begin{restatable}{problem}{ktransfersefx}[{\upshape$\efxopt$}]
\label{prob:k-valid-transfers-EFX}
   Given a fair division instance and an allocation $X$, determine the minimum number of valid transfers required to transform $X$ into an $\mathsf{EFX}$ allocation.
\end{restatable}

\begin{definition}[Near-$\mathsf{EFX}$ allocation]\label{def:near-efx}
An allocation $X$ of a set of goods is said to be near-$\mathsf{EFX}$ if exactly one \emph{fixed} agent is not $\mathsf{EFX}$-happy in $X$.
\end{definition}

 When agents have identical monotone valuations over items (that are either all goods or all chores), there always exists an $\mathsf{EFX}$ allocation \cite{PlautR20}. We shall now see that, unlike the case of $\mathsf{EF1}$, $\efxconn$ may not be possible, even when all agents have identical additive valuations over goods.

\begin{example} \label{example:identical-efx-counterexample}
   
We demonstrate a simple $\efxconn$ instance with $n = 2$ identical agents with additive valuations over $m = 6$ goods with an identical additive valuation $v$, and a near-$\mathsf{EFX}$ allocation in it, for which no valid transfer or valid exchange exists.

\begin{table}[htbp]
\centering
\setlength{\arraycolsep}{6pt}

\begin{tabular}{c@{\hspace{1.5em}}c}
$\displaystyle
\begin{array}{r|cccccc}
  & g_1 & g_2 & g_3 & g_4 & g_5 & g_6 \\
\hline
v(g) & \textcolor{violet}{1} & \textcolor{violet}{1} & \textcolor{violet}{1} & \textcolor{violet}{1} & \textcolor{teal}{6} & \textcolor{teal}{6}
\end{array}
$
&
$\displaystyle
\begin{array}{r|l}
\text{Agent} & \text{Items in } X_i \\
\hline
\textcolor{violet}{1} & \{g_1,\, g_2,\, g_3,\, g_4\} \\
\textcolor{teal}{2} & \{g_5,\, g_6\}
\end{array}
$
\end{tabular}

\caption{A near-$\mathsf{EFX}$ allocation $X$ for which no valid transfer or exchange exists.}
\label{tab:identical-efx-counterexample}
\end{table}

Clearly, the only $\mathsf{EFX}$-envy in $X$ is from agent $1$ towards agent $2$, so $X$ is near-$\mathsf{EFX}$. Finally, every transfer or exchange of items leads to new $\mathsf{EFX}$-envy being created, rendering all such operations invalid. Hence, $\efxconn$ may not be possible, even for identical additive valuations over goods. \hfill $\qed$

\end{example}

Finally, we show that even if we do not wish to maintain a near-$\mathsf{EFX}$ allocation after each operation, it may still take several transfers to restore $\mathsf{EFX}$.

\begin{example} \label{example:identical-efx-many-transfers}

Consider the following $\efxconn$ instance, with $n$ agents and $m$ items. Let $\G \ceq \set{g_1,\ldots,g_m}$.

Define the additive valuation $v$ as follows -- $v(g_1) = \ldots = v(g_{m-n}) \ceq 1$, and $v(g_{m-n+1}) = \ldots = v(g_m) \ceq m - n + 1$. Consider the near-$\mathsf{EFX}$ allocation $X$ defined as follows -- $X_1 \ceq \set{g_1,\ldots,g_{m-n}}$, $X_2 \ceq \set{g_{m-n+1},g_{m-n+2}}$, and for each $i \in \set{3,\ldots,n}$, let $X_i \ceq \set{g_{m-n+i}}$.

Clearly, the only $\mathsf{EFX}$-envy in $X$ is from agent $1$ towards agent $2$, so $X$ is near-$\mathsf{EFX}$. Note that the only way this can be repurposed into an $\mathsf{EFX}$ allocation is when all the goods valued $1$ (which currently belong to agent $1$) are all \enquote{equally distributed} among the $n$ agents, and agent $1$ receives one of the goods of value $m-n+1$ from agent $2$ - one may verify that this would involve making a really large number of transfers - at least $(n - 1)\left(\frac{m-n}{n}\right)$, to be precise. Furthermore, there is no way we can ensure near-$\mathsf{EFX}$ at any point during this process! \hfill $\qed$

\end{example}

 \cref{example:identical-efx-counterexample} and \cref{example:identical-efx-many-transfers} indicate that the $\efxconn$ problem needs operations more powerful than just transfers/exchanges to be permitted, in case one would like to maintain a near-$\mathsf{EFX}$ allocation after each such operation. We now state and prove the main result of this section.

\thmidenticalefx*
\label{thm:EFX-Restoration-weakly-NP-hard-identical-additive-restated}

\begin{proof}
We present our proof for the setting where we have all goods. The analogous result for the setting of all chores follows similarly (see Remark~\ref{rem:efx-chores}).

We begin by presenting a reduction from the \textsc{Partition} problem, defined below.

\begin{definition}[\textsc{Partition}]
Given positive integers $a_1,\ldots,a_m$, such that $\sum_{i=1}^m=2T$, decide whether there exists a set $R \subseteq [m]$ such that $\sum_{i \in R} a_i = T$.
\end{definition}

Consider an instance $\mathcal{S} \coloneq (a_1,\ldots,a_m;T)$ of \textsc{Partition}. Without loss of generality, we may assume that $1 \leq a_i < T$ for every $i \in [m]$, $\sum_{i=1}^m a_i > T$. 
The corresponding instance $\mathcal{I(S)}$ of $\efxconn$ is defined as follows:
There are $m+2$ agents denoted by $\{1,\ldots,m+2\}$, and $2m+2$ items $\{b_i,g_i\mid i\in [m]\}\cup\{p,c\}$, with valuations and the initial allocation as shown in \cref{tab:subset-sum-efx-reduction}. 
Note that the total value of all the items is $(m+3)T+\epsilon$. We refer to the items $g_i$, $i\in [m]$ as {\em $g$-items}.

\begin{table}[htbp]
\centering
\setlength{\arraycolsep}{8pt}
\renewcommand{\arraystretch}{1.15}

\[
\begin{array}{|r|l|c|}
\hline
\textsc{Agent} & \textsc{Initial bundle } X_i & \textsc{Value}\\
\hline
1 & \{p\} & \epsilon\\
\hline
i+1 \ (i\in[m]) & \{b_i,\, g_i\} & T+a_i\\
\hline
m+2 & \{c\} & T\\
\hline
\end{array}
\]

\caption{Constructed instance $\mathcal{I(S)}$ in the reduction from \textsc{Partition}. All agents have the same additive valuation function $v$. Here, $v(g_i) \coloneq a_i$.}
\label{tab:subset-sum-efx-reduction}
\end{table}

\noindent{\bf \textsc{Partition} $\implies$ $\efxconn$:}

We show that, if there is a positive solution to the \textsc{Partition} instance $\mathcal{S}$, then the constructed instance $\mathcal{I(S)}$ of $\efxconn$ also admits a positive solution. Let the solution to $\mathcal{S}$ be a subset $R$ of $\{a_1,\ldots,a_m\}$. Then, any sequence of transfers of the items $\set{g_i}_{i \in R}$ to agent $1$ is a sequence of valid transfers. Moreover, after transferring these items to agent $1$, the value of agent $1$'s bundle is $T+\epsilon$ whereas that of any other agent is at least $T$. Thus, the transfer sequence results in an $\mathsf{EFX}$ allocation. 

Now we show the other direction.\\

\noindent{\bf $\efxconn$ $\implies$ \textsc{Partition}:}

Now, we prove that a sequence of valid transfers that results in an $\mathsf{EFX}$ allocation corresponds to a subset of items in $\mathcal{S}$ that sum to value exactly $T$.
Let there be a sequence of valid transfers from $X$ to an $\mathsf{EFX}$ allocation $Y$ in $\mathcal{I(S)}$. We note the following properties of the sequence and that of $Y$:
\begin{claim}\label{clm:EFX-sequence}
The following statements hold.
\begin{enumerate}
    \item\label{itm:Y1} For each agent $i\in [m+2]$, $v(Y_i)\geq T$.
    \item\label{itm:size} There is an agent $i$ such that $|Y_i|=1$. Moreover, $v(Y_i)=T$.
    \item\label{itm:anchor} The anchor item $p\in Y_1$.
    \item\label{itm:H} No item of value $T$ is transferred to anyone.
\end{enumerate}
\end{claim}
\begin{proof}
{\bf \Cref{itm:Y1}: }If, for some $i$, $v(Y_i)<T$, then there must be an agent $j \neq i$ such that either $Y_j$ contains two items of value $T$ each or an item of value $T$ along with another item. In either of these cases, $i$ has an $\mathsf{EFX}$-envy towards
$j$ in $Y$, which is a contradiction.

{\bf \Cref{itm:size}:} This follows because there are $m+2$ agents and $2m+2$ items, and no agent can have an empty bundle in any $\mathsf{EFX}$ allocation. In fact, in all the intermediate allocations, no agent can have an empty bundle because the transfer leading to such a bundle cannot be a valid transfer. Hence, there exists an agent $i$ such that $|Y_i|=1$. Clearly, $v(Y_i)>T$ is not possible since no single item is valued strictly greater than $T$. Now, by \Cref{itm:Y1} above, $v(Y_i)<T$ is also not possible since $Y$ is an $\mathsf{EFX}$ allocation. Hence, $v(Y_i)=T$.

{\bf \Cref{itm:anchor}:} In all the allocations encountered before $Y$ in the transfer sequence, agent $1$ $\mathsf{EFX}$-envies someone. So the transfer of $p$ from $1$ to any other agent is not a valid transfer. 

{\bf \Cref{itm:H}:} Agent $1$ has $p$ throughout the transfer sequence, by \Cref{itm:anchor} above. So, if an agent $j$ transfers an item of value $T$ to agent $1$, $j$ will $\mathsf{EFX}$-envy $1$. So this cannot be a valid transfer. For the same reason, since every agent other than $1$ has an item of value $T$, none of them can receive an item of value $T$.
\end{proof}
From \cref{clm:EFX-sequence} above, $Y_1$ must consist of only $p$ and $g$-items. Moreover, we claim that $v(Y_1)=T+\epsilon$. This can be seen as follows: $v(Y_1)<T$ is not possible by \Cref{itm:Y1}. We cannot have $v(Y_1)=T$ as item values except that of $p$ are integral. If $v(Y_1)>T+\epsilon$ then, by \Cref{itm:size}, there is an agent $i\neq 1$ where $v(Y_i)=T$ who $\mathsf{EFX}$-envies agent $1$. This is because $v(Y_1)>T+\epsilon$ implies $v(Y_1\setminus\{p\})>T$. Thus the sequence of valid transfers must result in the transfer of a set $S$ consisting only $g$-items to agent $1$, and moreover, $v(S)=T$. This constitutes a 
solution to the partition instance $\mathcal{S}$. Hence, $\efxconn$ is weakly $\mathrm{NP}$-hard. One can also observe that this reduction also establishes weak $\mathrm{NP}$-hardness of $\efxopt$ -- set the valid transfer budget to $m/2$, and note that an $\mathsf{EFX}$ allocation is reachable from $\mathcal{I(S)}$ within $m/2$ valid transfers if and only if the \textsc{Partition} instance $\mathcal{S}$ has a positive solution.
This completes the proof.
\end{proof}
\begin{remark}\label{rem:efx-chores}
    We note that the above reduction, with minor modifications, gives $\mathrm{NP}$-hardness for $\efxconn$ with identical additive valuations over chores. The reduction for chores consists of the same instance as described above, with all values negated. Whenever, in the goods instance, agent $i$ $\mathsf{EFX}$-envies $j$, we have $\exists g\in X_j$ such that $v(X_i)<v(X_j\sm g)$. When the values are negated, we get a chore $c$ corresponding to the good $g$, and it holds that $\exists c\in X_j$ such that $v(X_i)>v(X_j\sm c)$. Thus, whenever $i$ $\mathsf{EFX}$-envies $j$ in the goods instance, $j$ $\mathsf{EFX}$-envies $i$ in the corresponding instance for chores. The same sequence of transfers then restores $\mathsf{EFX}$.
\end{remark}


\section{Future Directions} \label{sec:conclusion}

The meta-question raised in this work is the following: what is/are the \enquote{least complicated valid operation(s)} that one has to permit to reach a \emph{fair} allocation by never breaking the \emph{near-fair} guarantees? This leads to several natural open directions that one might explore; we mention a few here.

\begin{enumerate}
    \item \cref{example:identical-ef1-mixed-counterexample} shows that $\efconn$ may not be possible for mixed manna, even when the agents have identical additive valuations. However, it is not known whether $\efconn$ and $\efopt$ are $\mathrm{NP}$-hard in this setting. Due to the modified definition of $\mathsf{EF1}$ in this mixed manna setting, we may require a different style of reduction.
    \item We show that deciding whether $\efconn$ is possible under additive binary valuations is $\mathrm{NP}$-hard, whereas we give a polynomial-time algorithm when the valuations are graphical and the allocations are restricted to be orientations. Is it possible to remove one of these two restrictions, namely, the valuations being graphical, or allocations being orientations?
    \item We show $\mathrm{NP}$-hardness results for $\efconn$ in the setting where agents have additive binary valuations. Does the corresponding restoration problem for weaker notions of fairness like $\mathsf{PROP1}$ admit a polynomial-time algorithm in this setting?
\end{enumerate}


\section*{Acknowledgments}

P.\ N.\ would like to thank the Max Planck Institute for Informatics, SIC, Germany, for the invitation to visit them. A part of this work was done during this visit. A part of this work was done when N.\ R.\ was supported by the Lise Meitner Postdoctoral Fellowship (2023-25) at the Max Planck Institute for Informatics, SIC, Germany.


\bibliographystyle{alpha}
\bibliography{ref}

\end{document}